\documentclass[journal]{IEEEtran}

\usepackage{epsf,exscale,times}
\usepackage[english]{babel}
\usepackage{graphicx}
\usepackage{amsmath}
\usepackage{fancyhdr}
\usepackage{amsthm}
\usepackage{amsfonts}
\usepackage{amssymb}
\usepackage{bm}
\usepackage{bbm}
\usepackage{dsfont}
\usepackage{graphicx}
\usepackage{subcaption}
\usepackage{algorithm}
\usepackage{algorithmic}

\newcommand{\abs}[1]{\left|#1\right|}

\newcommand{\opn}[1]{\operatorname{#1}}

\newcommand{\norm}[1]{\left\|#1\right\|}

\newcommand{\ip}[2]{\langle #1, #2 \rangle}

\DeclareMathAlphabet{\bi}{OML}{cmm}{b}{it}
\DeclareMathAlphabet{\bcal}{OMS}{cmsy}{b}{n}
\DeclareMathAlphabet{\brmn}{OT1}{cmr}{bx}{n}

\DeclareMathSymbol{\R}{\mathalpha}{AMSb}{"52}

\newcommand{\brho}{\boldsymbol{\rho}}

\renewcommand{\L}{\langle}

\usepackage{amsthm}
\usepackage{amsfonts}

\def \x{\mathbf{x}}

\def \a{\mathbf{a}}

\def \d{\mathbf{d}}

\def \R{\mathbf{R}}

\def \L{\mathbf{L}}
\def \F{\mathbf{F}}

\def \0{\mathbf{0}}

\usepackage{mdframed}
\usepackage{color}
\usepackage{cite}
\usepackage{stfloats}
\usepackage{relsize}
\usepackage{amsopn}
\usepackage{float}
\restylefloat{table}

\usepackage{amsxtra}
\usepackage{amsmath}
\usepackage{graphicx}
\usepackage{latexsym}
\usepackage{amsfonts}
\usepackage{amssymb}
\usepackage{mathrsfs}
\usepackage{bm}
\usepackage{yhmath}
\usepackage{verbatim}
\usepackage{float}
\usepackage{xr-hyper}
\usepackage{multirow}
\usepackage{subcaption}

\usepackage{accents}

\usepackage{epstopdf}

\DeclareMathAlphabet{\bi}{OML}{cmm}{b}{it}
\DeclareMathAlphabet{\bcal}{OMS}{cmsy}{b}{n}
\DeclareMathAlphabet{\brmn}{OT1}{cmr}{bx}{n}

\usepackage{amsthm}
\usepackage{amsfonts}

\graphicspath{{./images/}}

\DeclareMathAlphabet{\bi}{OML}{cmm}{b}{it}
\DeclareMathAlphabet{\bcal}{OMS}{cmsy}{b}{n}
\DeclareMathAlphabet{\brmn}{OT1}{cmr}{bx}{n}

\def \x{\mathbf{x}}

\def \f{\mathbf{f}}

\newcommand{\mrm}[1]{\mathrm{#1}}

\newcommand{\thmref}[1]{\textbf{Theorem~\ref{#1}}}

\newcommand{\lemref}[1]{\textbf{Lemma~\ref{#1}}}
\newcommand{\propref}[1]{\textbf{Proposition~\ref{#1}}}
\newcommand{\corref}[1]{\textbf{Corollary~\ref{#1}}}

\newcommand{\figref}[1]{\textbf{Figure~\ref{#1}}}

\newcommand{\assref}[1]{\textbf{Assumption~\ref{#1}}}

\newtheorem{theorem}{Theorem}
\newtheorem{lemma}{Lemma}
\newtheorem{proposition}{Proposition}
\newtheorem{assumption}{Assumption}
\newtheorem{corollary}{Corollary}
\newtheorem{definition}{Definition}
\theoremstyle{remark}
\newtheorem{remark}{Remark}

\begin{document}
\date{\vspace{-7ex}}
\title{Exact Imaging of Extended Targets Using Multistatic Interferometric Measurements}
\author{\IEEEauthorblockN{Bariscan Yonel \IEEEauthorrefmark{1}\IEEEauthorrefmark{2}, \IEEEmembership{Student Member, IEEE}, Il-Young Son, \IEEEauthorrefmark{1}\IEEEauthorrefmark{3},
    and Birsen Yazici\IEEEauthorrefmark{1}\IEEEauthorrefmark{4}, \IEEEmembership{Senior Member, IEEE}}
\thanks{\IEEEauthorrefmark{4}Corresponding author.}
\thanks{\IEEEauthorrefmark{1}Yazici, Son, and Yonel are with the Department of Electrical, Computer and Systems Engineering,
Rensselaer Polytechnic Institute, 110 8th Street, Troy, NY 12180 USA, E-mail:
\IEEEauthorrefmark{4}yazici@ecse.rpi.edu, \IEEEauthorrefmark{2}yonelb@rpi.edu, \IEEEauthorrefmark{3}soni@rpi.edu
Phone: (518)-276 2905, Fax: (518)-276 6261.}
\thanks{This work was supported by the Air Force Office of Scientific Research
(AFOSR) under the agreement FA9550-16-1-0234, Office of Naval Research
(ONR) under the agreement N0001418-1-2068 and by the National Science
Foundation (NSF) under Grant No ECCS-1809234.}}
\IEEEtitleabstractindextext{
  \begin{abstract}
    In this paper, we present a novel approach for exact recovery of extended targets
    in wave-based multistatic interferometric imaging, based on Generalized Wirtinger Flow (GWF)
    theory~\cite{bariscan2018}.  Interferometric imaging is a
    generalization of phase retrieval, which arises from cross-correlation of measurements from
    pairs of receivers.
  Unlike
    standard Wirtinger Flow, 
    GWF theory
    guarantees exact recovery for arbitrary lifted forward models that satisfy the restricted isometry property over rank-1, positive semi-definite (PSD) matrices with a sufficiently small restricted isometry constant (RIC).
    To this end, we design a deterministic, lifted forward model for interferometric multistatic radar satisfying the exact recovery conditions of the GWF theory.
Our results quantify a lower limit on the pixel spacing and the minimal sample complexity for exact multistatic radar imaging.   
We provide a numerical study of our RIC and pixel spacing bounds, which shows that GWF can achieve exact
recovery 
with super-resolution.
While our primary interest lies in radar imaging, our method is also applicable to other
  multistatic wave-based imaging problems such as those arising in acoustics and geophysics.
\end{abstract}
}
\maketitle
\IEEEdisplaynontitleabstractindextext

\section{Introduction}
\subsection{Motivation and Objective}

In this paper, we study the exact reconstruction of complex scenes in the context of
  multistatic interferometric imaging.
Interferometric imaging is a close relative of phaseless imaging where, in lieu of
self-correlated, intensity only data, we have pairwise cross-correlated data that introduces a
phase component.
This work establishes \emph{Generalized Wirtinger Flow} (GWF), a computationally
  efficient interferometric imaging
method developed in~\cite{bariscan2018}, as a theoretical framework for exact multistatic imaging
of complex scenes, while
relating its recovery guarantees to the imaging system parameters.
To this end, we design a
deterministic and underdetermined measurement model 
satisfying the GWF's sufficient condition for
exact recovery. In addition, we show that it is possible to obtain exact reconstruction at
resolutions smaller than Fourier-based methods and provide minimum order of measurements sufficient to guarantee such reconstruction.

The recently developed GWF algorithm is inspired by standard Wirtinger Flow
(WF)~\cite{candes2015phase} developed for the generalized phase retrieval problem. WF is a
computationally efficient alternative to lifting based methods~\cite{Candes13a, Candes13b}.
The GWF algorithm extends the standard WF to
interferometric inversion problems, and identifies a sufficient condition for exact recovery for
arbitrary measurement models, characterized over the \emph{lifted} domain.
Hence, unlike standard WF theory which guarantees exact recovery for specific random measurement
models, GWF theory 
guarantees exact recovery for a general class of inverse
problems including random and deterministic models that abide by a single condition.
In particular, the sufficient condition requires the lifted forward map to satisfy the restricted
isometry property (RIP) for rank-$1$, positive semi-definite (PSD) matrices with a sufficiently
small restricted isometry constant (RIC).
To the best of our
knowledge, our work is the first  in which a deterministic and underdetermined forward model
satisfying RIP for rank-$1$ PSD matrices in the lifted domain has been designed. 

We provide two outcomes that unify the imaging problem with the mathematical theory of GWF.
First, we determine the minimum pixel spacing to satisfy the sufficient condition for exact recovery via the GWF algorithm.
Our lower bound depends on the imaging system parameters, thereby, quantifies the range of values and imaging scenarios for exact recovery guarantees to hold.
For common radar imaging parameters spanning passive and active imaging modalities, this
fundamental lower bound outperforms the range resolution limit of 
Fourier-based imaging methods
for sufficiently small scenes.

Secondly, we determine the sample complexity in the order of the number of unknowns to be
reconstructed.
Unlike the classical results from electromagnetic scattering theory which study the degrees of freedom of scattered fields via their spatially band-limited nature and Nyquist sampling \cite{bucci1987spatial, bucci1989degrees}, 
our analysis is based on the discrete problem, with a geometry consisting of sparsely distributed, static, terrestrial receivers. 
Furthermore, our sampling complexity result directly relates to the exact recovery guarantees of our algorithm through the properties of the lifted forward map, rather than the accuracy of interpolating the sampled scattered field. 
Notably, our result establishes that the exact recovery guarantees of GWF hold for a lifted forward model that is underdetermined.  
Hence, we specify a multistatic measurement model with an optimal order
of measurements for interferometric wave-based imaging. 

\subsection{Related Work and Advantages of GWF}

Interferometric techniques have a rich history in acoustic, geophysical, and electromagnetic wave-based imaging.  
Cross correlations are frequently deployed as a fundamental formulation for passive modalities \cite{Yarman08, Garnier09, wang10, ammari2013passive, mason15}, in which the received ambient signal originates from a source of opportunity, such as a wireless communication signal, digital TV signal or FM radio for radar imaging. 
On the other hand, cross-correlations were proposed for active imaging problems to mitigate the effects of statistical fluctuations in scattering media~\cite{garnier2005imaging, lobkis2001emergence}, clutter \cite{borcea2003theory, borcea2005interferometric, borcea2006coherent}, and phase errors in the correlated linear transformations~\cite{blomgren2002super, gough2004displaced, flax1988phase, mason2016robustness}.
More recently, inversion from phaseless scattered fields were proposed \cite{novikov2015formulas, bardsley2016kirchhoff, chen2017phaseless}, to fully eliminate the need of coherent data acquisition in various modalities to cut implementation costs \cite{novikov2015illumination}, or evade fundamental issues in maintaining phase coherence over long synthetic apertures \cite{laviada2014phaseless}. 
Here, we motivate our approach for interferometric multi-static radar imaging via its advantages over several conventional and modern methods in the literature. 

\subsubsection{Passive Radar}
A popular method for passive imaging is the time difference of arrival (TDOA) backprojection~\cite{Yarman08, wang10, Wang11, wang12, LWang12, Wacks14}.
Although they are computationally efficient, TDOA backprojection is based on certain assumptions on
the scatterers~\cite{Yarman08} that are not applicable for realistic scenes and can produce undesirable
background artifacts~\cite{Kissinger2012}.

To mitigate this problem, methods based on lifting have
been
adapted to the interferometric measurement model for passive imaging~\cite{mason15, Demanet13}.
These methods are inspired by the convex semi-definite programming approaches to phase retrieval~\cite{Candes13a, Candes13b, waldspurger2015phase, Chai11}, which reformulate inversion as a low rank matrix recovery
(LRMR) problem. 
Convexification has the added advantage that LRMR is known to
have theoretical exact recovery guarantees under certain conditions on the lifted forward
map~\cite{Recht10}.
However, these
advantages come at the cost of increasing the dimension of the inverse problem, and hence introduce
several 
limitations.
Specifically, as a result of lifting, LRMR suffers from limitations on spatial sampling of the imaging grid due to high computational complexity and demanding memory requirements~\cite{mason15, bariscan2018}.

GWF is a non-convex optimization approach that operates fully on the original signal domain, thus avoids lifting the problem at implementation and provides computational and memory efficiency over LRMR methods.
Unlike TDOA/FDOA backprojection, GWF guarantees exact recovery without additional prior knowledge or limiting assumptionson the scene.
Furthermore, the exact recovery guarantees afforded by LRMR~\cite{Recht10} require more stringent conditions on the lifted forward map than that of
GWF~\cite{bariscan2018}.

\subsubsection{Active Radar}
For active imaging, there exists a rich literature of methods on general multistatic
geometries 
involving distributed
antennas~\cite{haimovich2008,chernyak98} or arrays~\cite{Levari00}.
These include
time
reversal and beamforming~\cite{griffiths82,PRADA1994}, subspace methods (MUSIC~\cite{Levari00,arogyaswami93}, linear sampling~\cite{Kirsch98,Kirsch2000,arens2009linear}), and iterative optimization schemes~\cite{strong2003edge,yu2010,herman09}.

Time-reversal,
beamforming and
MUSIC have wide use in array imaging problems.
These methods assume 
that the scatterers in the scene of interest are point-like and 
the number of measurements are greater than number of scatterers  in the scene~\cite{arogyaswami93,kim02}.  
This is in stark contrast to our GWF framework, in which no such assumptions are needed.

Linear sampling methods were devised 
to extend the applicability of subspace methods to the reconstruction of extended targets in the
far field and can recover the boundaries of extended objects~\cite{Kirsch98,Kirsch2000}.
Similar to our imaging system geometry, linear sampling
methods consider a scenario that the receivers and transmitters fully encircle the scene of
interest in the far field.
However, the method degrades considerably when the aperture angle is less than $2\pi$ radians~\cite{audibert17}.
Our GWF result quantifies the impact of the aperture angle directly in the sufficient condition for exact recovery.
Hence, GWF has applicability when aperture angle to the scene is limited.


Regularized iterative reconstruction approaches, such as total variation (TV)~\cite{strong2003edge} and $\ell$-1 regularization~\cite{yu2010,herman09},
have shown to achieve edge
preservation.  However, regularized iterative reconstruction approaches, in general, do not offer a
theoretical exact recovery guarantee. 
Notably, TV regularization, while convex, is known to have multiple non-trivial minimizers.
In addition, the TV regularizer does not have a closed form proximity operator, hence iterative reconstruction requires an inner optimization problem at each iteration.
Similar problems also exist with $\ell_1$ regularization due to existence of a tuning parameter,
which is heuristically determined.
More importantly, the $\ell_1$ regularizer is based on strong sparsity assumptions on the scene, which is not applicable to realistic scenes.
GWF, on the other hand, offers exact recovery guarantees for complex, realistic 
scenes
with low computational complexity per iteration.

\subsection{Organization of the Paper}
In Section~\ref{sec:signal}, we describe the signal model for interferometric multistatic radar. 
Section~\ref{subsec:mainres} presents our
main results which establish how imaging system parameters of multistatic radar
controls the RIC for rank-$1$ real-valued PSD matrices.
Section~\ref{sec:numsim}
describes the simulated experiments performed to verify our results in Section~\ref{subsec:mainres}.
Section~\ref{sec:conc} concludes the paper.


\section{Signal Model}\label{sec:signal}

\subsection{Received Data Model}\label{subsec:signal1}
Let $N$ be the number of receivers each deployed at different spatial locations $\a_i^r$,
$i=1,\dots,N$, 
where subscript $i$ denotes the $i$-th receiver 
and superscript $r$ denotes receiver.  Assume a
single transmitter located at $\a^t$.
Furthermore, without loss of generality, we assume that the ground topography is flat.
Thus, each spatial location in three-dimensional space is represented as $\x=[\bm{x}, 0]$ where $\bm{x}\in\mathbb{R}^2$.
Under these assumptions and that only the scattered field is being measured under the Born approximation, the received signal at each receiver for multistatic radar can be
modeled as ~\cite{cheney2009fundamentals}
\begin{align}
  f_i(\omega) = &\int_{D} \mrm{e}^{\mrm{i}\omega/c_0
    \phi_i(\x)}A_i(\bm{x},\omega)\rho(\bm{x}) d\bm{x},\label{eq:di2}\\
  &\omega\in[\omega_c-B/2, \omega_c+B/2]\subset \mathbb{R}\nonumber
\end{align}
where
\begin{equation}
  \phi_i(\x) = \abs{\x-\a_i^r} + \abs{\x-\a^t}\label{eq:phi__}
\end{equation}
is the bistatic phase function, $D\subset\mathbb{R}^2$ is the support of the scene,
$\omega$ is the fast-time frequency variable,
$\omega_c$ is the center frequency, $B$ is the bandwidth,
$c_0$ is the speed of light, $\rho$ is the target/scene reflectivity function; 
and $A_i$ is the amplitude function given by
\begin{equation}
  A_i(\x,\omega) = \frac{J_i(\x,\omega) J_t (\x, \omega)}{\abs{\x-\a_i^r}\abs{\x-\a^t}}
  \label{eq:Aij}
\end{equation}
with $J_i$, and $J_t$ 
being the receiver and transmitter antenna beampatterns. 
\subsection{Correlated Measurements}\label{subsec:signal2}
Given the data model~\eqref{eq:di2}, we consider the interferometric data, i.e. fast-time
cross-correlation of the measurements at pairs of different receivers. 
Furthermore, we make the assumption that $\abs{J_t(\x,\omega_m)} \approx C_t \in \mathbb{R}^{+}$. 
In other words, we assume 
that the
transmitted waveform has a flat spectrum, and that the scene remains in the $-3$ dB beam-width of the illumination pattern. 
This is typical of radar waveforms and waveforms of opportunity such as
phase shift keying (PSK) modulation found in orthogonal frequency-division multiplexing (OFDM)
common among digital communications. 
Using~\eqref{eq:di2}--\eqref{eq:Aij}, the correlated measurements can be modeled as
\begin{equation}
  d_{i,j}(\omega) = \int_{D\times D}
  \mrm{e}^{\mrm{i}\omega/c_0\varphi_{i,j}(\x,\x')}A_{i,j}(\bm{x},\bm{x}', \omega)
  \tilde{\rho}(\bm{x},\bm{x}') d\bm{x} d\bm{x}'\label{eq:dij}
\end{equation}
where
\begin{equation}
  \varphi_{i,j}(\x,\x') = \abs{\x-\a_i^r}+\abs{\x-\a^t}-\abs{\x'-\a_j^r}-\abs{\x'-\a^t},\label{eq:varphi}
\end{equation}
\begin{equation}\label{eq:corrramp}
  A_{i,j}(\bm{x},\bm{x}', \omega) = A_i(\bm{x}, \omega)A_j^{*}(\bm{x}', \omega)
\end{equation}
and
\begin{equation}
  \tilde{\rho}(\bm{x},\bm{x}') = \rho(\bm{x})\rho^{*}(\bm{x}')
\end{equation}
with $(\cdot)^*$ denoting complex conjugation.
We call $\tilde{\rho}$ the lifted version of $\rho$ or the \emph{Kronecker scene}.

We next make the small-scene and far-field approximation and approximate the phase term in~\eqref{eq:varphi} 
as
\begin{equation}
  \varphi_{i,j}(\x,\x') \approx \abs{\a_i^r}-\abs{\a_j^r}-\ip{\hat{\a}_i^r}{\x} + \ip{\hat{\a}_j^r}{\x'}-\ip{\hat{\a}^t}{\x-\x'}\label{eq:varphi2}
\end{equation}
where $\hat{\a}$ is the unit vector in the direction of $\a$, and ~\eqref{eq:corrramp} as 
\begin{equation}
  A_{i,j}(\bm{x},\bm{x}', \omega)\approx \alpha_{i,j}:=\frac{C_t^2 {J_i} (\bm{x}_o, \omega) {J_j}(\bm{x}_o, \omega)^*}{\abs{\a_i^r}\abs{\a_j^r}\abs{\a^t}^2}\label{eq:Aij2}
\end{equation}
where $\bm{x}_o$ denotes the center of the scene, with $\abs{J_i (\bm{x}_o, \omega)} \approx C_i$ over the observed frequency band. 
We assume that the support of the scene is discretized into $K$ discrete spatial points,
$\{\bm{x}_k|\ k=1,\dots,K\}$ and define $\bm{\rho} = [\rho(\bm{x}_1),\dots,\rho(\bm{x}_K)]^T$.
We further assume that the support of $\omega$ is discretized into $M$ samples, $\Omega =
\{\omega_m| m=1,\dots, M\}$
so that $\bm{d}_{i,j} = [d_{i,j}(\omega_1),\dots,d_{i,j}(\omega_M)]^T$,
$\omega_m = \omega_c - B/2 + \frac{m-1}{M}B$.

We write~\eqref{eq:dij} as
\begin{equation}
  d_{i,j}(\omega_m) = 
  \ip{\L_{i}^m}{\bm{\rho}}\ip{\L_j^m}{\bm{\rho}}^{*} = \opn{tr}\left(\L_j^m(\L_i^m)^H\tilde{\bm{\rho}}\right)\label{eq:dij3}
\end{equation}
where
\begin{equation}
  \L_i^m = [\mrm{e}^{-\mrm{i}\omega_m/c_0 \phi_i(\x_k)}A_i]_{k=1}^K, \qquad i=1,... N.\label{eq:Li}
\end{equation}

Let
\begin{equation}
  \d = \frac{1}{\sqrt{M\binom{N}{2}}}[\bm{d}_{1,2}^T,\dots, \bm{d}_{N-1,N}^T]^T
\end{equation}
be the full vectorized data scaled by the number of correlated measurements.~\eqref{eq:dij3}
shows that the data vector $\d$ is linear in $\tilde{\bm{\rho}}$,
the  Kronecker scene, while it is non-linear in $\bm{\rho}$.
Thus, the data vector can be written as
\begin{equation}
  \d = \mathcal{F}(\tilde{\bm{\rho}})\label{eq:datamodel}
\end{equation}
where $\mathcal{F}$ is a linear mapping from $\mathbb{R}^{K\times K}$ to $\mathbb{C}^{M\binom{N}{2}}$.
Alternatively, if $\overline{\bm{\rho}}$
is the column-wise vectorization of $\tilde{\bm{\rho}}$, 
\begin{equation}
  \d = \F\overline{\bm{\rho}}
\end{equation}
where $\F$ is a complex-valued matrix of size $M\binom{N}{2}\times K^2$, whose rows are formed by row-wise
vectorization of the matrix $\L_i^m(\L_j^m)^H$.
\section{Exact Multistatic Wave-Based Imaging}\label{subsec:mainres}

In this section, we are concerned with identifying the imaging system parameters, i.e.,
design of the measurement vectors $\L_i^m$, $i=1,\dots,N$, $m=1,\dots,M$, so that
the lifted forward map $\mathcal{F}$ satisfies the
sufficient condition proved 
in~\cite{bariscan2018} for exact recovery via the GWF algorithm.
We establish all the results presented in this section under the following assumption.
\begin{assumption}\label{assump:2}
  Let
  \begin{equation}
    \begin{aligned}
    \Phi_{i,j}^{k,k',l,l'} =&
    \ip{\hat{\a}_i^r}{\x_k-\x_{k'}}-\ip{\hat{\a}_j^r}{\x_l-\x_{l'}}\\
    &+\ip{\hat{\a}^t}{\x_k-\x_{k'}-\x_l+\x_{l'}}.
    \end{aligned}
    \label{eq:phasefunc}
  \end{equation}
  Then, we assume that $\frac{B}{2Mc_0}\Phi_{i,j}^{k,k',l,l'}\ll 2\pi$ for all
  $(i,j,k,k',l,l')$   where $B$ is the bandwidth of the received signal, $M$ is the number of frequency samples, and
  $c_0$ is the speed of light in a vacuum.
\end{assumption}
\assref{assump:2} is used to make small angle approximation in the proof of~\lemref{lem:lem1} below.
This assumption implies that the number of frequency
samples needed depends on the
bandwidth of the transmitted waveform
and the maximum value of $\Phi_{i,j}^{k,k',l,l'}$, 
which depends on the
size of the scene and the placement of the receivers.
As later seen in~\eqref{eq:Mcomp}, this assumption is easily satisfied if the scene is sufficiently small. 

We next introduce the following lemma
that expresses the kernel of the operator
$\mathcal{F}$ in terms of sinc functions.  This lemma is used in proving \textbf{Propositions}
\textbf{\ref{prop:prop1}} and \textbf{\ref{prop:prop2}}, and for the main result in~\thmref{thm:Theorem1}.
\begin{lemma}\label{lem:lem1}
  Suppose~\assref{assump:2} holds.
  Then, the $2$-norm of the data, $\d$ can be written as
  \begin{equation}
    \begin{aligned}
      \norm{\d}_2^2 = \norm{\mathcal{F}\tilde{\brho}}^2_2
      = \frac{\sum_{i<j}|\alpha_{i,j}|^2}{\binom{N}{2}}&\left(\norm{\tilde{\brho}}_F^2 +
        \sum_{k\neq  k',l\neq l'}\mathcal{K}(\Phi_{i,j}^{k,k',l,l'})\right.\\
        &\left. \times\tilde{\rho}(\bm{x}_k,\bm{x}_{k'})\tilde{\rho}(\bm{x}_{l'},\bm{x}_{l}) \right)
  \end{aligned}\label{eq:splits3}
  \end{equation}
  where the phase term $\Phi_{i,j}^{k,k',l,l'}$ is as in~\eqref{eq:phasefunc}
  and
  \begin{equation}
    \mathcal{K}(\Phi) =\frac{\sin\left[\left(\omega_c'+\frac{B}{2}\right)\frac{\Phi}{c_0}\right] -
      \sin\left[\left(\omega_c'-\frac{B}{2}\right)\frac{\Phi}{c_0}\right]}{B\frac{\Phi}{c_0}},\label{eq:K}
  \end{equation}
  with $\omega_c' = \omega_c - \frac{B}{2M}$.
\end{lemma}
\begin{proof}
  See Appendix~\ref{prf:lem1}.
\end{proof}

\subsection{GWF Framework for Interferometric Imaging}
For establishing the GWF as an exact interferometric wave-based imaging framework, we study the restricted isometry property (RIP) of the lifted forward map $\mathcal{F}$ over the set of rank-$1$, PSD matrices. 
\begin{definition}\label{def:RIP}
Let $\mathcal{F} : \mathbb{C}^{K \times K} \rightarrow \mathbb{C}^{M\binom{N}{2}}$ denote the lifted forward model provided in \eqref{eq:datamodel}. Then $\mathcal{F}$ satisfies the restricted isometry property over rank-1, positive semi-definite matrices of the form $\tilde{\brho} = \brho \brho^H$, with a restricted isometry constant (RIC)-$\delta$, if
\begin{equation}
(1 - \delta ) \| \tilde{\brho} \|_F^2 \leq \frac{1}{M} \| \mathcal{F} (\tilde{\brho} ) \|^2 \leq (1 + \delta ) \| \tilde{\brho} \|_F^2,
\end{equation}
for any $\brho \in \mathbb{C}^K$, where $\| \cdot \|_F$ denotes the Frobenius norm. 
\end{definition}

Notably, we consider a domain of $\brho \in \mathbb{R}^K$ in optimization via GWF, for which the exact recovery guarantees of \cite{bariscan2018} directly apply. 
Thereby, the GWF algorithm for interferometric imaging is summarized as follows:
\begin{itemize}
\item \textbf{Inputs:} Interferometric measurements $d_{ij}(\omega_m)$ and measurement vectors $\mathbf{L}_i^m$, $i = 1, \cdots N$. 

\item \textbf{Initialization:} \emph{Backproject} the interferometric measurements to the lifted domain, i.e., form an estimator for the Kronecker scene as: 
\begin{equation}\label{eq:SpecInit}
\hat{\mathbf{X}} = \frac{1}{M} \mathcal{P}_{S} \left(\mathrm{Re} \{\mathcal{F}^H ( \d ) \right)\},
\end{equation}
and keep its rank-1 approximation $\lambda_0 \brho_0 \brho_0^H$, where $\mathcal{P}_S$ is the projection operator onto the set of symmetric matrices. The initialization step consists of the outer product of the two measurement vectors for each of the $M$ samples, resulting in $\mathcal{O}(MK^2)$ multiplications, followed by an eigenvalue decomposition with $\mathcal{O}(K^3)$ complexity.

\item \textbf{Iterations:} Initializing with $\brho_0$, perform gradient descent updates as $\brho_{k+1} = \brho_{k} - \frac{\mu_{k}}{\| \brho_0 \|^2} \nabla \mathcal{J}(\brho_{k})$, having 
\begin{equation}\label{eq:objgwf}
\mathcal{J}(\brho) = \frac{1}{2} \| \mathcal{F}(\brho \brho^H) - \mathbf{d} \|_2^2,
\end{equation}
which yields
\begin{equation}\label{eq:gradgwf}
\nabla \mathcal{J}( \brho_k ) = \frac{1}{M} \mathcal{P}_S \left( \mathrm{Re} \{\mathcal{F}^H(\mathbf{e}_k)\} \right) \brho_{k},
\end{equation}
where $(\mathbf{e}_k)_{ij}^m = (\left(\mathbf{L}_i^m \right)^H \brho_{k} \brho_{k}^H \mathbf{L}_j^m - d_{ij}(\omega_m) ), \ \forall i \neq j$, $m = 1, \cdots, M$. 

Each iteration requires the following operations:
\begin{enumerate}
\item Computing and storing the linear terms $(\mathbf{L}_{i,j}^m)^H \brho_k$, requiring $M$ number of $K$ multiplications for each, resulting in $\mathcal{O}(MK)$ multiplications. 
\item Computing the error by cross correlating linear terms, requiring $\mathcal{O}(M)$ multiplications.
\item Multiplication of the linear terms $(\mathbf{L}_{i,j}^m)^H \brho_k$ and the error $e_{ij}^m$ for each $m = 1, \cdots M$, requiring $\mathcal{O}(M)$ multiplications.
\item Multiplication of the result in $3$ with vectors $\{\mathbf{L}_i^m\}_{m=1}^M$ and $\{\mathbf{L}_j^m\}_{m=1}^M$, requiring $\mathcal{O}(MK)$ multiplications.
\end{enumerate}
These operations result in $\mathcal{O}(MK)$ multiplications for each iteration. 
\end{itemize}

In particular, \cite{bariscan2018} establishes that exact recovery of a ground truth $\brho_t \in \mathbb{C}^K$ is
guaranteed upto a global phase factor, if the forward operator for
the lifted Kronecker scene, satisfies the RIP over the set of rank-$1$, 
PSD matrices (i.e., $\tilde{\brho}$) with RIC of less than $0.214$. 
Furthermore, starting from the initial estimate computed from ~\eqref{eq:SpecInit}, gradient descent iterations minimizing ~\eqref{eq:objgwf} using ~\eqref{eq:gradgwf} converges geometrically to the true solution at a rate $1 - \kappa$, where $\kappa$ is upper bounded by
\begin{equation}\label{eq:convRate}
\frac{2 \mu}{\alpha} \leq \frac{(1-\delta_1) h^2(\delta_1)}{(1+\delta_1) c^2(\delta_1)}, 
\end{equation}
where $h \leq c$ are positive constants solely depending on $\delta$ \cite{bariscan2018}. 



\subsection{Asymptotic Result}

As a stepping stone for our main result, we begin by showing the asymptotic isometry of $\mathcal{F}$ defined in~\eqref{eq:datamodel}, as $\omega_c\rightarrow\infty$ and $N\rightarrow\infty$.
Following our asymptotic analysis of the kernel of $\mathcal{F}$, we characterize its RIP over rank-1, PSD matrices in the non-asymptotic regime.
As a result of our non-asymptotic analysis, we derive an upper bound on the restricted isometry constant that is controlled by the imaging system parameters.

Despite its limited use in practice, our initial asymptotic result 
offers a valuable benchmark for the non-asymptotic case.
Notably, it justifies assessing how the isometry of $\mathcal{F}$ is perturbed over the set of rank-1, PSD matrices when the central frequency $\omega_c$, and the number of receivers $N$ are finite.
We specifically make use of this perspective in establishing our main result, by analytically evaluating elements $\brho \brho^H$ in the range of $\mathcal{F}^H \mathcal{F}$.
Furthermore, it characterizes the expected limiting behavior of our upper bound estimate on the RIC-$\delta$.

The following proposition shows that in the asymptotic regime, i.e., as $\omega_c$ gets large,
$\mathcal{F}$ becomes a delta function
with respect to the phase term $\Phi_{i,j}^{k,k',l,l'}$.
\begin{proposition}\label{prop:prop1}
  Under~\assref{assump:2}, we have
  \begin{equation}
  \lim_{\omega_c'\rightarrow\infty} \mathcal{K}(\Phi_{i,j}^{k,k',l,l'}) = \begin{cases}
    0 & \Phi_{i,j}^{k,k',l,l'}\neq 0\\
    1 & \Phi_{i,j}^{k,k',l,l'}= 0.
  \end{cases}
\end{equation}
\end{proposition}
\begin{proof}
  See Appendix~\ref{prf:prop}.
\end{proof}
Given~\propref{prop:prop1},
the next proposition shows that in the limit as $\omega_c\rightarrow\infty$ and
$N\rightarrow\infty$, $\mathcal{F}$ is an isometry.
\begin{proposition}[Asymptotic Isometry of $\mathcal{F}$ for large $\omega_c$ and $N$]\label{prop:prop2}
  Under~\assref{assump:2}, we have
  \begin{equation}
    \begin{aligned}
    \lim_{\omega_c'\rightarrow\infty,
      S\rightarrow\infty}&\frac{1}{\binom{N}{2}}\sum_{i<j}\abs{\alpha_{i,j}}^2 W_{i,j}\\
    =&\frac{1}{\binom{N}{2}}\sum_{i<j}\abs{\alpha_{i,j}}^2 \sum_{k\neq k',l\neq
      l'}\mathcal{K}(\Phi_{i,j}^{k,k',l,l'})\\
    &\qquad\times\tilde{\rho}(\bm{x}_k,\bm{x}_{k'})\tilde{\rho}(\bm{x}_{l'},\bm{x}_{l})
    = 0
  \end{aligned}
  \end{equation}
\end{proposition}
\begin{proof}
  See Appendix~\ref{prf:prop2}.
\end{proof}
Since in the asymptotic regime $\mathcal{F}$ is an isometry, we can deduce that the RIC over rank-$1$, PSD
should become small as
$\omega_c$ and $N$ get large.  This motivates us to find an upper bound on the rank-$1$, PSD RIC
constant in the non-asymptotic regime in terms of the imaging parameters.  In the next subsection,
we establish this upper bound.


\subsection{Non-asymptotic Result}
Before we introduce our main theorem, we introduce two further assumptions.
\begin{assumption}\label{assump:square}
  The scene is enclosed by a square with side $L$ and sampled regularly on a square grid.  The
  coordinate system is centered at the middle of the square.  Hence,
  $\bm{x} = [x_1,x_2]^T\in [-L/2,L/2]\times [-L/2,L/2]$ with $\sqrt{K}$ samples in both $x_1$- and
  $x_2$-axis and $L=\sqrt{K}\Delta$ where $\Delta$ is the pixel spacing.
\end{assumption}


Under \assref{assump:square}, it is easy to see that the phase term $|\Phi_{i,j}^{k,k',l,l'}|$ is upper bounded by $4L \sqrt{2}$ for any selection of $i,j,k,k',l,l'$.
{
Then, for \assref{assump:2}, letting $\Delta_{res}= 2\pi \frac{c_0}{2B}$ be the range resolution given by the Fourier-based methods
the small angle approximation holds to high accuracy if}
\begin{equation}\label{eq:Mcomp}
  M\geq\mathcal{O}\left(\frac{L}{\Delta_{res}}\right),
\end{equation}
since $\max_{i,j,k,k',l,l'} |\Phi_{i,j}^{k,k',l,l'}| =
\mathcal{O}\left(L\right)$. For instance, $M \geq 5.8 \frac{L}{\Delta_{res}}$ corresponds to a $<1\%$ error for the $\mathrm{sinc}$ approximations in~\lemref{lem:lem1}.

\begin{assumption}\label{assump:config}
  \begin{enumerate}
  \item The receivers are isotropic and lie on a circular arc equidistant from each other and to the center of the
    coordinate system.
    Let $A\in(0,2\pi]$ be the aperture of the multistatic system.  Then, the azimuth angles of the
    look-directions are multiples of $A/N$.
  \item 
    All receivers and the transmitter 
    are located at the same height. 
    Let $\phi$ be the elevation angle in radians.
    Then, $\hat{\a}^r_i = [\cos\phi\cos\theta_i,\cos\phi\sin\theta_i,\sin\phi]$ where
    $\theta_i=\frac{A i}{N}$, $i=0,\dots,N-1$ are the azimuth angles of the receivers' look-directions.
  \item The transmitter 
    is located on the
    $x_1$-axis.  Hence, $\hat{\a}_t = [\cos\phi,0,\sin\phi]^T$.
\end{enumerate}
\end{assumption}

\noindent\assref{assump:config} allows us to make integral approximation to a Riemann sum in the
proof of~\thmref{thm:Theorem1} (see Appendix~\ref{prf:Theorem1}).  The
approximation error is then incorporated into the result of~\thmref{thm:Theorem1}.
Note that the assumption on the location of the transmitter is
not essential, but is there for convenience.

We now state our non-asymptotic result in the following theorem, which establishes an upper bound on the rank-$1$, PSD RIC for the data
model presented in~\eqref{eq:datamodel}, in terms of the underlying imaging parameters.
\begin{theorem}[RIC of the Lifted Forward Mapping of Multistatic Imaging]\label{thm:Theorem1}
  Let
  \begin{equation}\label{eq:defsthm1}
    \lambda_c = \frac{2\pi c_0}{\omega_c'}
  \end{equation}
  be the wavelength corresponding to the center frequency.
Then, under \textbf{Assumptions} \textbf{\ref{assump:square}}, \textbf{\ref{assump:config}}, and~\lemref{lem:lem1}, we have the
  following upper bound on the restricted isometry constant $\delta$ of $\mathcal{F}$ over rank-1, PSD matrices:
  \begin{equation}
    \delta \leq {\frac{2\pi}{A}}\frac{2\lambda_c\sqrt{L\Delta_{Res}}}{\Delta^2\cos\phi\sqrt{\cos\phi}} +
		\mathcal{O}\left(\frac{K}{(N/A)^2} {\lambda_c^{-{3}/{2}}}  \right)
    \label{eq:delta}
  \end{equation}
  where the order is a small constant and
  \begin{equation}\label{eq:defsthm2}
    \Delta_{res} = 2\pi\frac{c_0}{2B},\quad \text{and}\quad     \Delta = \frac{L}{\sqrt{K}}.
  \end{equation}
\end{theorem}
\begin{proof}
  See Appendix~\ref{prf:Theorem1}.
\end{proof}
As provided in ~\eqref{eq:convRate} and explained in~\cite{bariscan2018}, $\delta$ directly controls the convergence
rate of GWF iterates.
As such,
bound in~\eqref{eq:delta} establishes that the convergence behavior of GWF for multistatic imaging
depends on system parameters such as
the center frequency $\omega_c$, the bandwidth $B$, the number of receivers $N$, the number of
unknowns $K$, as well as the side length $L$ of the scene.

\begin{remark}
Observe that $N$ has a higher order than $\lambda_c^{-1}$ in the second term in \eqref{eq:delta}. 
Hence, our RIC upper bound estimate tends to $0$ as $\omega_c \rightarrow \infty$, $N \rightarrow \infty$, consistent with our asymptotic isometry result for $\mathcal{F}$.
Specifically, the first term in~\eqref{eq:delta} captures the perturbation from the limit when the central frequency is finite, whereas the second term characterizes the perturbation due to having finite number of receivers. 
In fact, the second term directly arises from the closed form error of a Riemann sum approximation to an integration  over look directions of the receivers.
\end{remark}

\begin{remark}
The Riemann sum error behaves in an inverted manner to the first term with respect to the central frequency of the transmitted signal is increased, given a fixed imaging aperture and number of look directions. 
This is indeed an expected outcome, as the data collection manifold corresponds to a larger area of the 2D Fourier spectrum of the scene as the central frequency is increased while the imaging aperture is fixed. 
As a result, $N$ number of look directions corresponds to a poorer discretization of the data collection manifold, and the factor of $\lambda_c^{-3/2}$ in the second term in \eqref{eq:delta} relates directly to this phenomenon. 
\end{remark}

Using the decoupled nature of our upper bound estimate on the RIC, we quantify the minimal pixel spacing at which the exact recovery guarantees of GWF can hold. 
\begin{corollary}[Resolution Bound]\label{cor:Cor1}
  Suppose we have sufficiently many receivers, i.e., $N^2\gg K$, such that the second term in~\eqref{eq:delta} is negligible. Then GWF guarantees exact
  recovery if
  \begin{equation}
    \Delta\geq \sqrt{{\frac{2\pi}{A}}\frac{2\lambda_c\sqrt{L \Delta_{res}}}{0.214\cos\phi\sqrt{\cos\phi}}}.\label{eq:bound2}
  \end{equation}
\end{corollary}
\begin{proof}
  Assuming $N^2\gg K$, the second term in~\eqref{eq:delta} in the upper bound of $\delta$ vanishes. Recall that exact recovery is
  guaranteed via GWF if $\delta$ is less than or equal to $0.214$. Upper bounding the RIC bound in~\eqref{eq:delta}, we have
  \begin{equation}
    {\frac{2\pi}{A}}\frac{2\lambda_c\sqrt{L \Delta_{res}}}{\Delta^2\cos\phi\sqrt{\cos\phi}}\leq 0.214.
    \label{eq:exact}
  \end{equation}
  The rest follows by rearranging~\eqref{eq:exact}.
\end{proof}
Notably, even with $N \rightarrow \infty$, \eqref{eq:bound2} is the absolute best resolution at which exact multi-static imaging is possible by GWF.
Hence, Corollary~\ref{cor:Cor1} yields a fundamental bound for the pixel spacing in designing realizable imaging systems with finite number of receivers.
\begin{figure*}[!ht]
  \centering
  \captionsetup[subfigure]{justification=centering}
  \begin{subfigure}[t]{0.33\textwidth}
    \centering
    \captionsetup[subfigure]{justification=centering}
    \includegraphics[width=\textwidth]{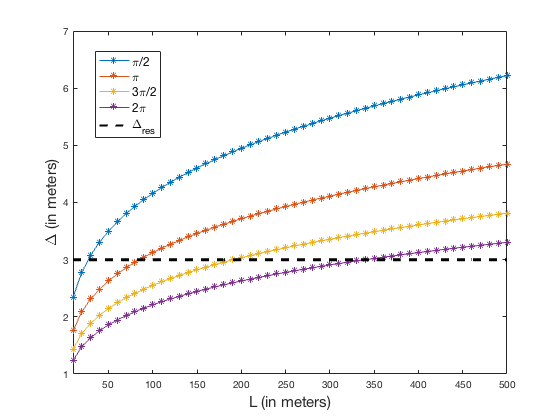}
    \caption{Active Regime. Center frequency was set at $10$ GHz and bandwidth at $50$ MHz.}\label{fig:pixbnd_active}
  \end{subfigure}\quad
  \begin{subfigure}[t]{0.33\textwidth}
    \centering
    \includegraphics[width=\textwidth]{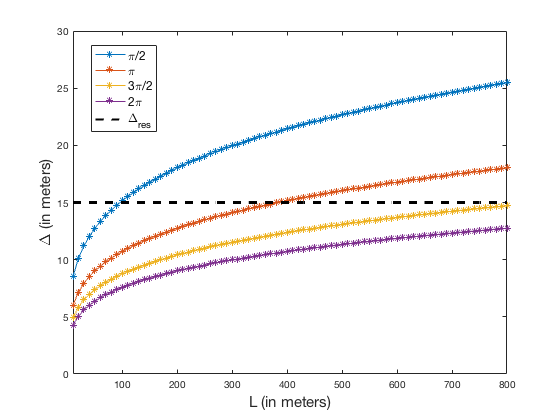}
    \caption{Passive Regime.  Center frequency was set at $1.9$ GHz and bandwidth at $10$ MHz
      (similar to CDMA cell phone signals).}\label{fig:pixbnd_passive}
  \end{subfigure}
  \caption{Curves of lower bound on the pixel spacing, $\Delta$ for various values of aperture
    lengths at active and passive regimes.}\label{fig:pixlobnd}
\end{figure*}

The resolution bound of Corollary~\ref{cor:Cor1} corresponds to the super-resolution regime when reconstructing small scenes in both active, and passive scenarios, as depicted in~\figref{fig:pixbnd_active} and~\figref{fig:pixbnd_passive}, respectively.
Note that as $L$ gets large, the lower bound eventually becomes greater than the range resolution limit of the Fourier-based methods.
This is in agreement with our theoretical arguments, which are established under a small scene approximation.
It should also be stressed that our lower bound abides by the sufficient condition for exact recovery, but it is not a necessary one.
Therefore, while recovery of scenes at a higher resolution than $\Delta_{res}$ may still be possible via GWF, it is not covered by the theory in \cite{bariscan2018}.

Additionally, the sufficient number of receivers for~\eqref{eq:bound2} to hold is $N^2 \geq \mathcal{O}(K)$.
Since $M = \mathcal{O}(L)$ by~\eqref{eq:Mcomp}, this implies that super-resolution imaging via GWF requires a sample complexity of at least $MN^2 = \mathcal{O}(K^{3/2})$.
We reduce this complexity result by the following corollary, which quantifies the minimal sample requirement for exact multi-static imaging via GWF at a fixed pixel spacing that abides the lower bound of \corref{cor:Cor1} .

\begin{corollary}[Sample Complexity]\label{cor:samp}
Given the final result of~\thmref{thm:Theorem1}, exact multistatic imaging condition for GWF is satisfied at the following sample complexity:
  \begin{equation}
    MN^2=\mathcal{O}(K^{5/4}).
  \end{equation}
\end{corollary}
\begin{proof}
Reorganizing the upper bound on $\delta$ in~\thmref{thm:Theorem1}, we have
\begin{equation}
  c_1 \frac{K}{L \sqrt{L}} + c_2 \frac{K}{N^2} = \tilde{\delta}
\end{equation}
where $c_1$, $c_2$ are $\mathcal{O}(1)$ as functions of $K$ and $N$.
Now, for any fixed pixel spacing $\Delta$, we have $L=\mathcal{O}(\sqrt{K})$.  Thus,
\begin{equation}
  \tilde{c}_1 K^{1/4} + \tilde{c}_2 \frac{K^{3/4}}{N^2} K^{1/4} = \tilde{\delta}\label{eq:K14}
\end{equation}
for some $\tilde{c}_1, \tilde{c}_2 = \mathcal{O}(1)$.
Observe that $K^{1/4}$ factor in the first term of the left-hand side
of~\eqref{eq:K14} is non-vanishing and hence at best yields the RIC upper bound of $\tilde{\delta} =
\mathcal{O}(K^{1/4})$.
Now from~\assref{assump:2}, we have $M = \mathcal{O}(L)$.  Thus, the minimal
sample complexity in which the RIC upper bound is in the order of $K^{1/4}$ is achieved when $N^2 =
\mathcal{O}({K}^{3/4})$.
Therefore,
\begin{equation}\label{eq:delta2}
  (\tilde{c}_1 + \hat{c}_2)K^{1/4} = \tilde{\delta}
\end{equation}
when $MN^2 = \mathcal{O}(K^{5/4})$.
\end{proof}

In addition to the minimal sample complexity, ~\corref{cor:samp} yields a rate at which the algorithm performance deteriorates.
Clearly, from \eqref{eq:delta2}, our ability to fine sample the scene while attaining the exact recovery guarantees of GWF for multi-static imaging depends on the dimension of the problem, at a rate $K^{1/4}$, or equivalently, $\sqrt{L}$.
This, again, is consistent with our theoretical arguments as we derive our results through a small scene approximation.

The fact that the upper bound of $\delta$ has a non-vanishing $K^{1/4}$ factor reveals an interesting phenomenon that is also observed in the performance of spectral initialization in phase retrieval literature, even when the measurement vectors are random.
This degradation with the increasing dimension of the unknown is not captured in the probabilistic analysis with random measurement vectors, yet is indeed a significant issue which forms the basis for sample truncation in computing the initialization and gradient estimates \cite{chen2017solving}.

Specifically for deterministic, wave-based multistatic imaging problems, ~\corref{cor:samp} necessitates a system design such that the controllable constants in \eqref{eq:delta2} sufficiently suppress the $K^{1/4}$ factor.
This promotes GWF as a highly applicable method in passive imaging scenarios where the range resolution is limited, or in active imaging scenarios where small, isolated extended targets are being imaged, with possible extensions and applications in spot-light mode synthetic aperture radar \cite{yonel2018phaseless}.

\section{Numerical Simulations}\label{sec:numsim}
% ...
In this section, we provide several numerical simulations demonstrating veracity of the theory
presented in Section~\ref{subsec:mainres}.
The following multistatic set-up is common to all simulations presented in this section, and
conforms to the assumptions laid  out in Section~\ref{subsec:mainres}.
\begin{enumerate}
\item There is a single transmitter located at $[15.8, 0, 0.25]$ km.
\item The transmitted waveform has unit amplitude frequency spectrum.
\item Varying number of receivers are distributed equidistant on an arc of a circle of radius $10$ km from the
  scene center at a height of $0.25$ km.
\item The scene of interest is square with flat topography.
\end{enumerate}


\figref{fig:config} illustrates the multistatic set-up used in this section. Note that the
  illustration is not to scale.
\begin{figure}[!ht]
  \centering
  \includegraphics[width=0.3\textwidth]{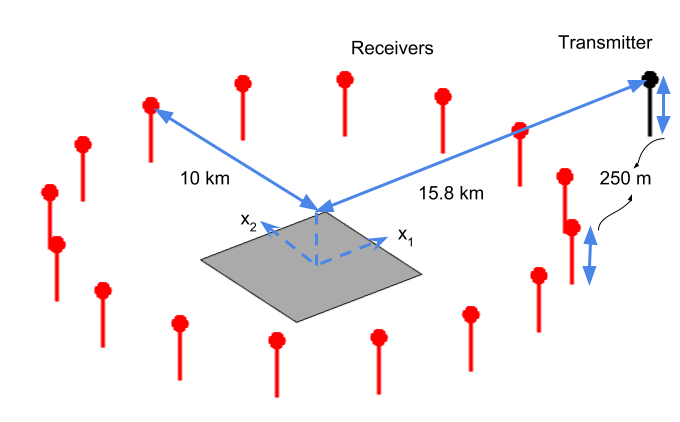}
  \caption{Illustration of the multistatic imaging set-up for numerical simulations. 
    (Not to scale.)}\label{fig:config}
\end{figure}

The figure-of-merit we use throughout is the mean square error (MSE) of the reconstructed scene.  This
is computed by taking the per pixel difference between the true scene and the reconstructed scene and
averaging the squares of the differences.

In all experiments presented, our data is synthetically generated using our received signal model in ~\eqref{eq:di2} under the single-scattering assumption. 
In Sections~\ref{subsec:norec},~\ref{subsec:bw},
and~\ref{subsec:wc}, a single parameter is varied in each set of experiment while all other relevant parameters are
fixed.  The parameters are chosen in the active and passive imaging ranges.~\figref{fig:true_scene}
shows the scene used for all experiments in the subsequent sections. 
Finally, in Section~\ref{subsec:noise}, we provide simulations that depict the performance of GWF in non-ideal conditions, namely under additive noise at low signal-to-noise ratios (SNR), and discretization mismatch between the ground truth, and the reconstructed image. 
\begin{figure}[!ht]
  \captionsetup[subfigure]{justification=centering}
  \begin{subfigure}[t]{0.23\textwidth}
    \centering
  \includegraphics[width=\textwidth]{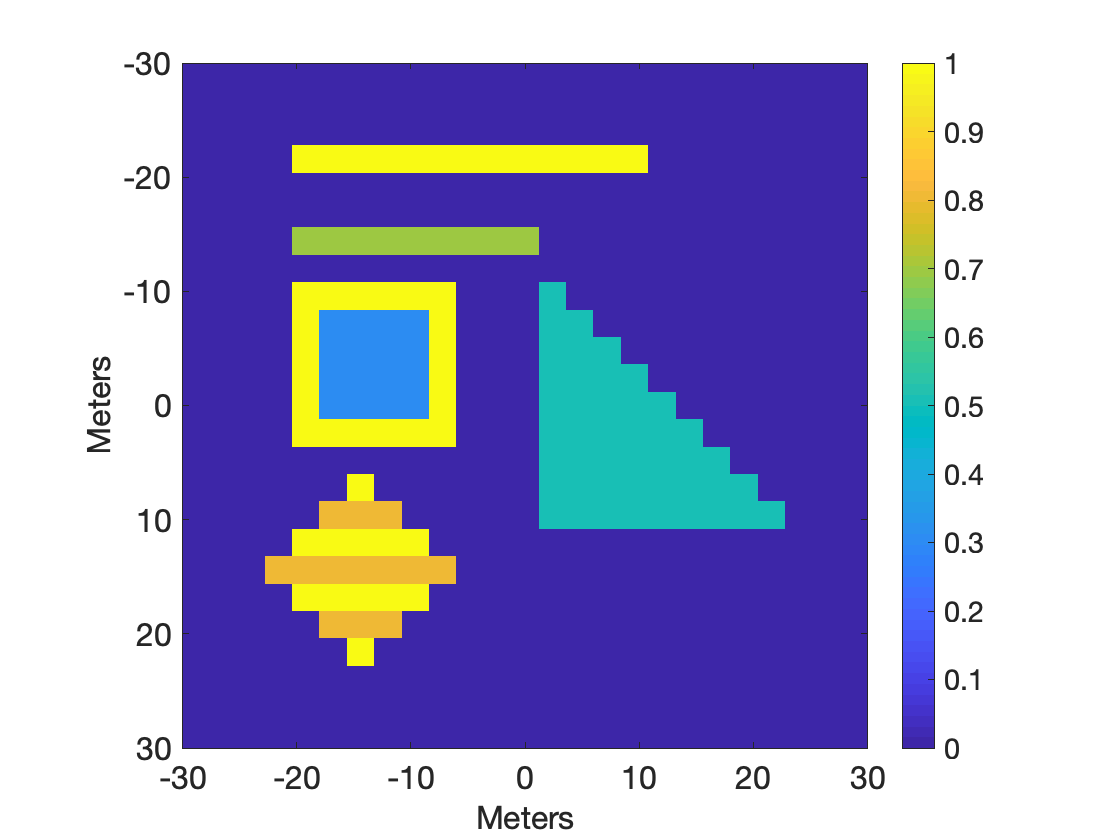}
    \caption{Active case.}\label{fig:act_gt}
  \end{subfigure}\quad
  \begin{subfigure}[t]{0.23\textwidth}
    \centering
  \includegraphics[width=\textwidth]{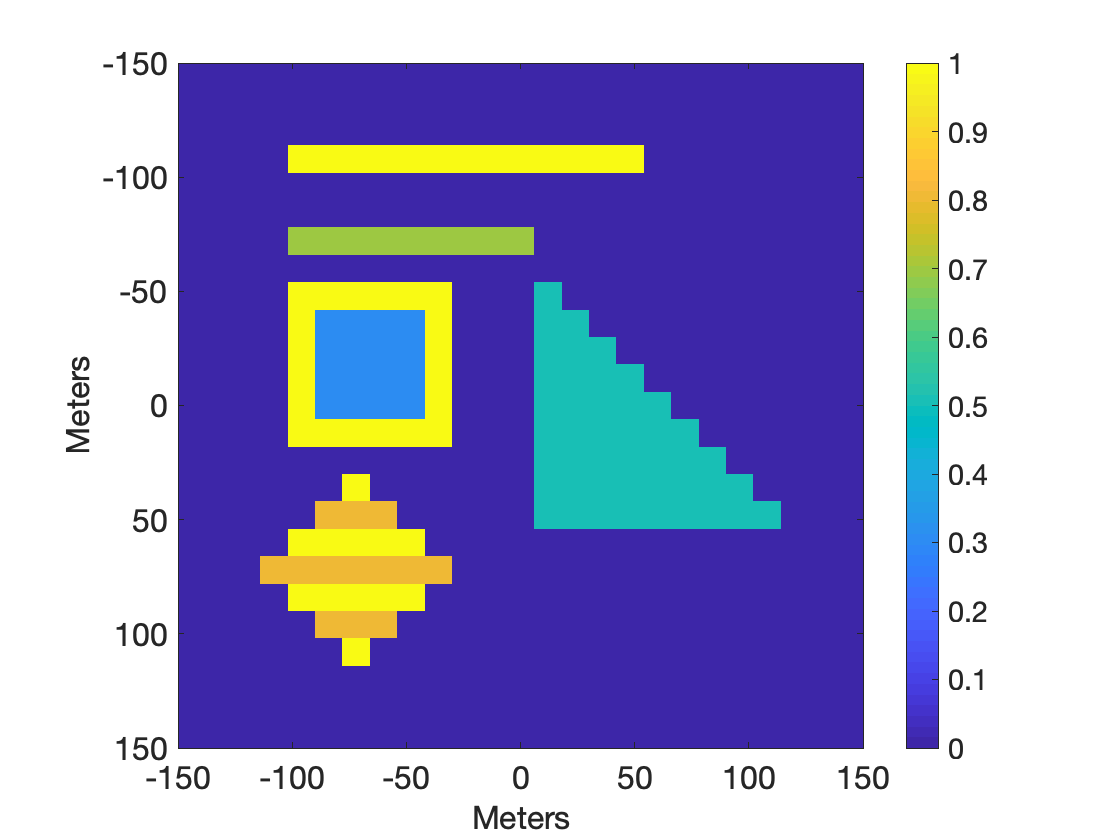}
    \caption{Passive case.}\label{fig:pass_gt}
  \end{subfigure}
  \caption{The ground truth used in the numerical experiments. The colorbar refers to the reflectivity of the resolution bins. $L$ is set as $60$ m in the active, $300$ m in the passive case.}\label{fig:true_scene}
\end{figure}
\subsection{Effect of Number of Receivers}\label{subsec:norec}
The first series of numerical experiments are designed to verify the effect of the number of receivers
on the performance of GWF reconstruction.
In~\eqref{eq:delta}, the second term
involves the square of the number of receivers, $N^2$, in the denominator.  Thus, we expect the
number of receivers to have significant effect on the quality of the reconstruction.  To verify the
effect of the number of
receivers on the reconstruction, we ran a series of simulations with varying number of
receivers while fixing all other relevant parameters in active or passive radar regimes.
\begin{figure}[!h]
  \centering
  \includegraphics[width=0.33\textwidth]{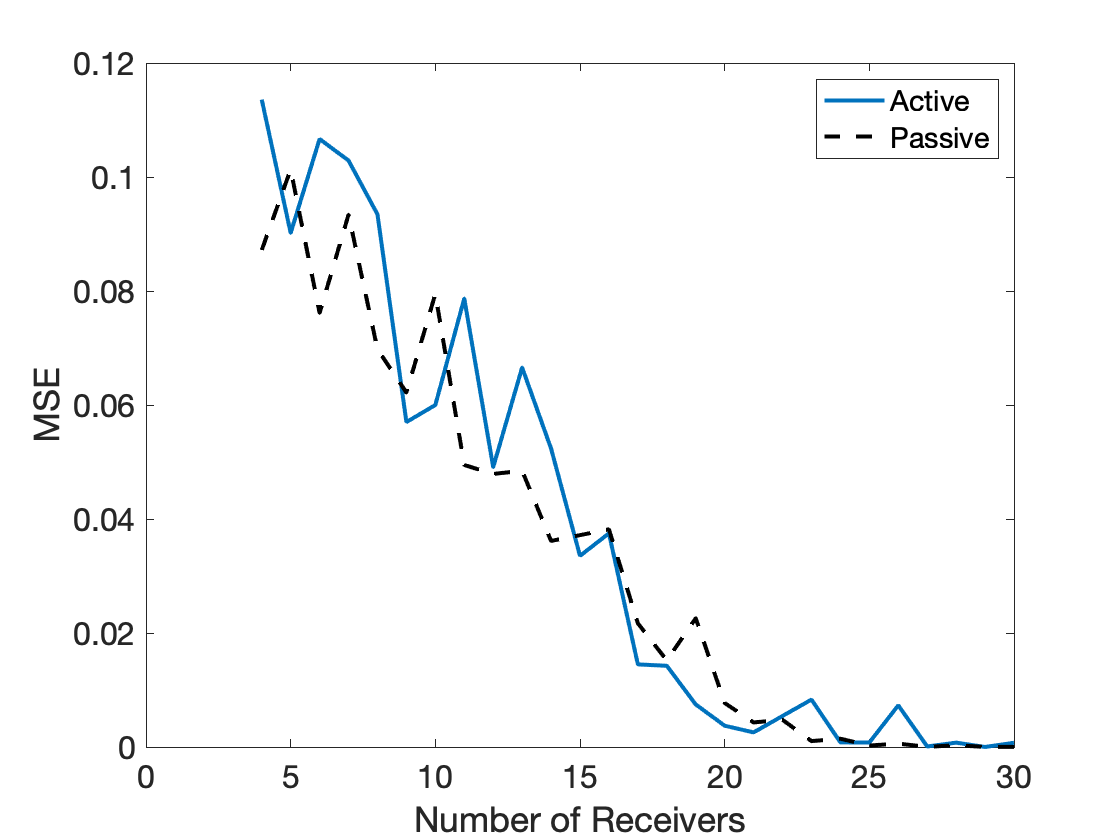}
  \caption{Number of receivers vs. MSE of the reconstruction after $4000$ iterations of GWF for
    active and passive radar parameters.  Blue solid line is the curve for active radar parameters and
    black dashed line is for the passive radar parameters.  Number of frequency samples was held constant at $64$ and
    $K=625$ for
    both cases. The pixel spacing was set at $2.4$m for active case and $12$m for passive.  The
    center frequency was set at $10$ GHz and $1.9$ GHz for active and passive cases, respectively.
    The bandwidth was set at $50$ MHz and $10$ MHz for active and passive cases, respectively.
  }\label{fig:norec}
\end{figure}

\begin{figure}[!ht]
  \centering
  \captionsetup[subfigure]{justification=centering}
  \begin{subfigure}[t]{0.23\textwidth}
    \centering
    \includegraphics[width=\textwidth]{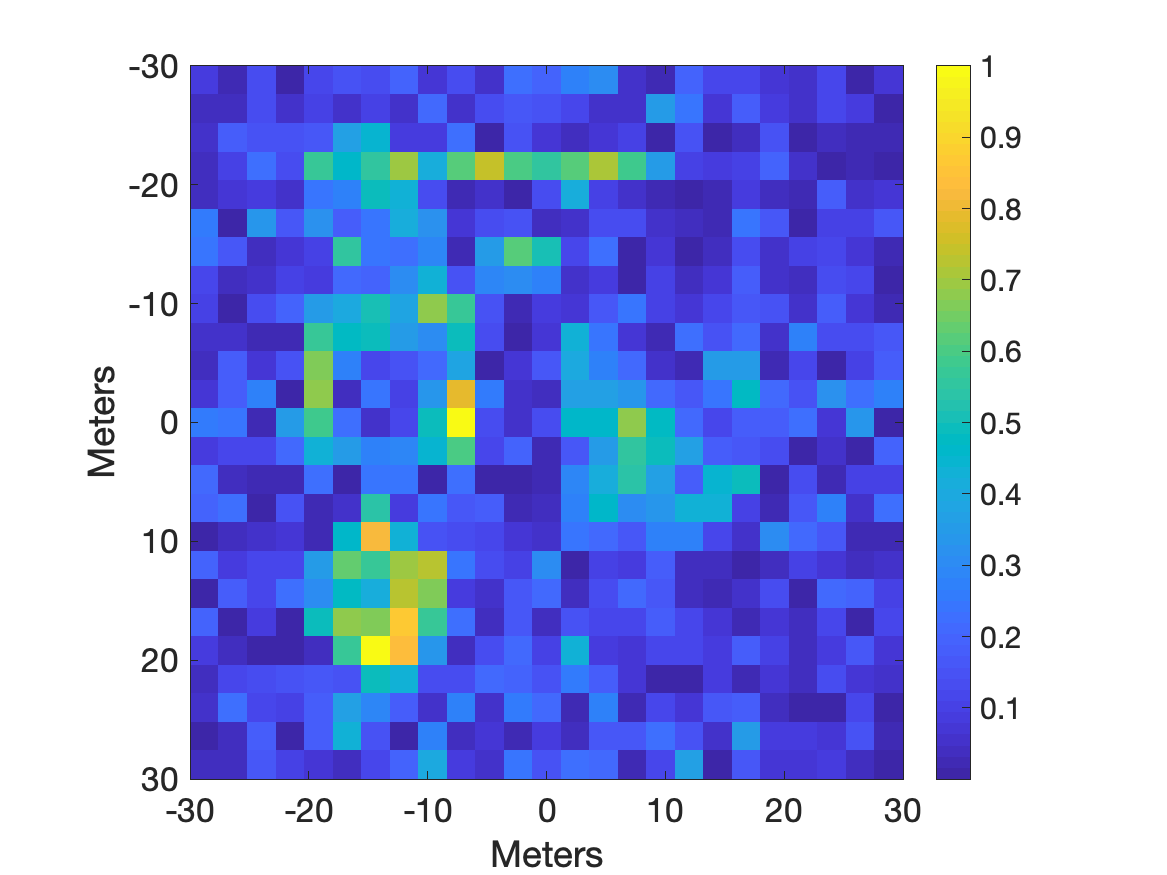}
    \caption{$12$ receivers.}\label{fig:12rec_active}
  \end{subfigure}\quad
  \begin{subfigure}[t]{0.23\textwidth}
    \centering
    \includegraphics[width=\textwidth]{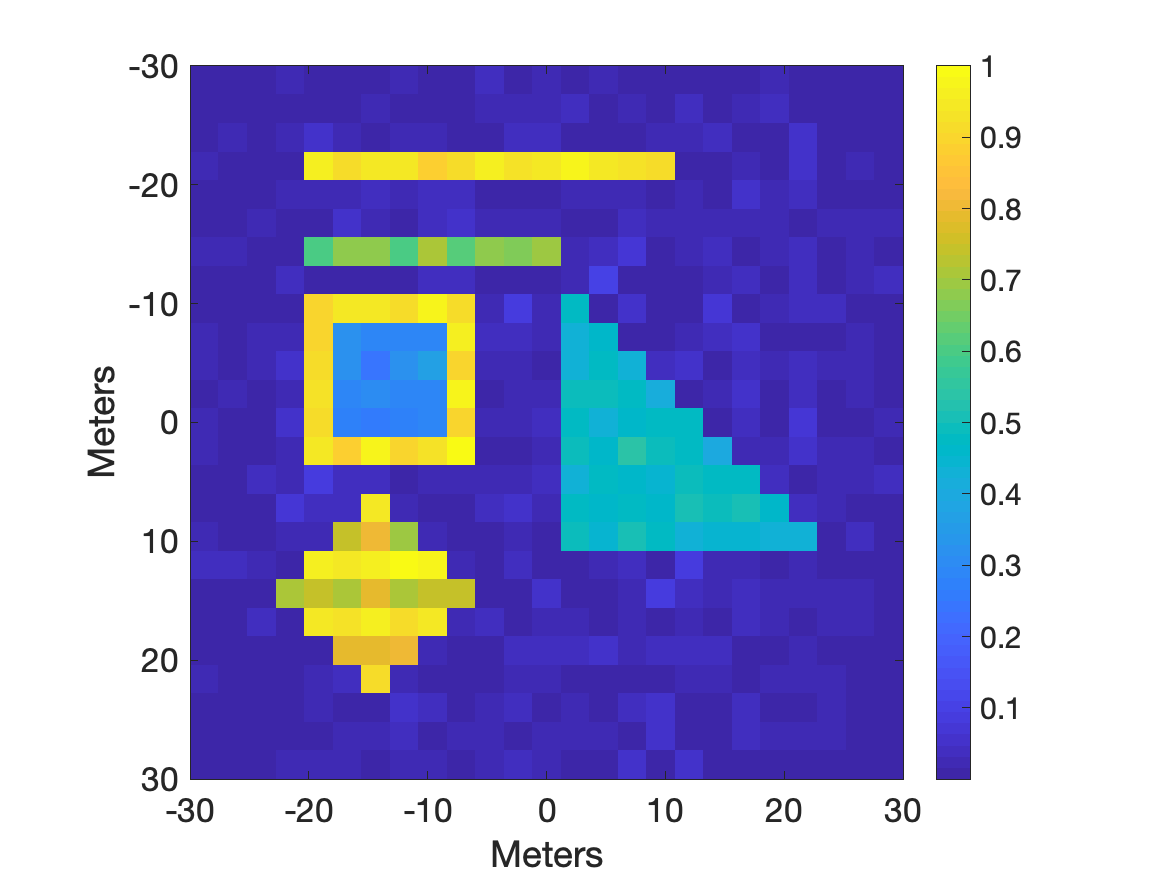}
    \caption{$24$ receivers.}\label{fig:24rec_active}
  \end{subfigure}
  \caption{Sample reconstructions after $4000$ iterations of GWF for active imaging case with varying number of receivers.
    Bandwidth was set at $50$ MHz with center frequency of $10$ GHz.   Number of frequency samples was held constant at $64$ and
    $K=625$. 
    The pixel spacing was set at $2.4$ m.
  }\label{fig:norec_recons_active}
\end{figure}

\begin{figure}[!ht]
  \centering
  \captionsetup[subfigure]{justification=centering}
  \begin{subfigure}[t]{0.23\textwidth}
    \centering
    \includegraphics[width=\textwidth]{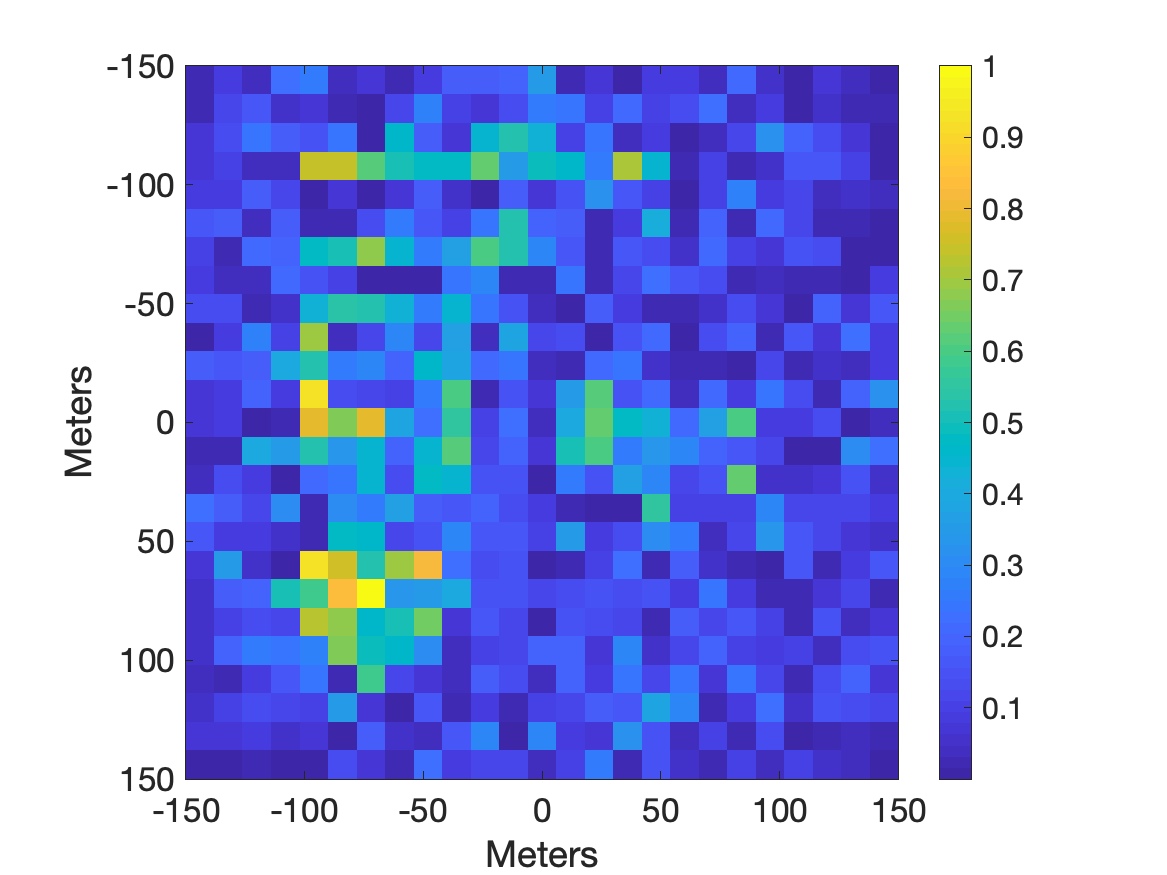}
    \caption{$12$ receivers.}\label{fig:12rec_passive}
  \end{subfigure}\quad
  \begin{subfigure}[t]{0.23\textwidth}
    \centering
    \includegraphics[width=\textwidth]{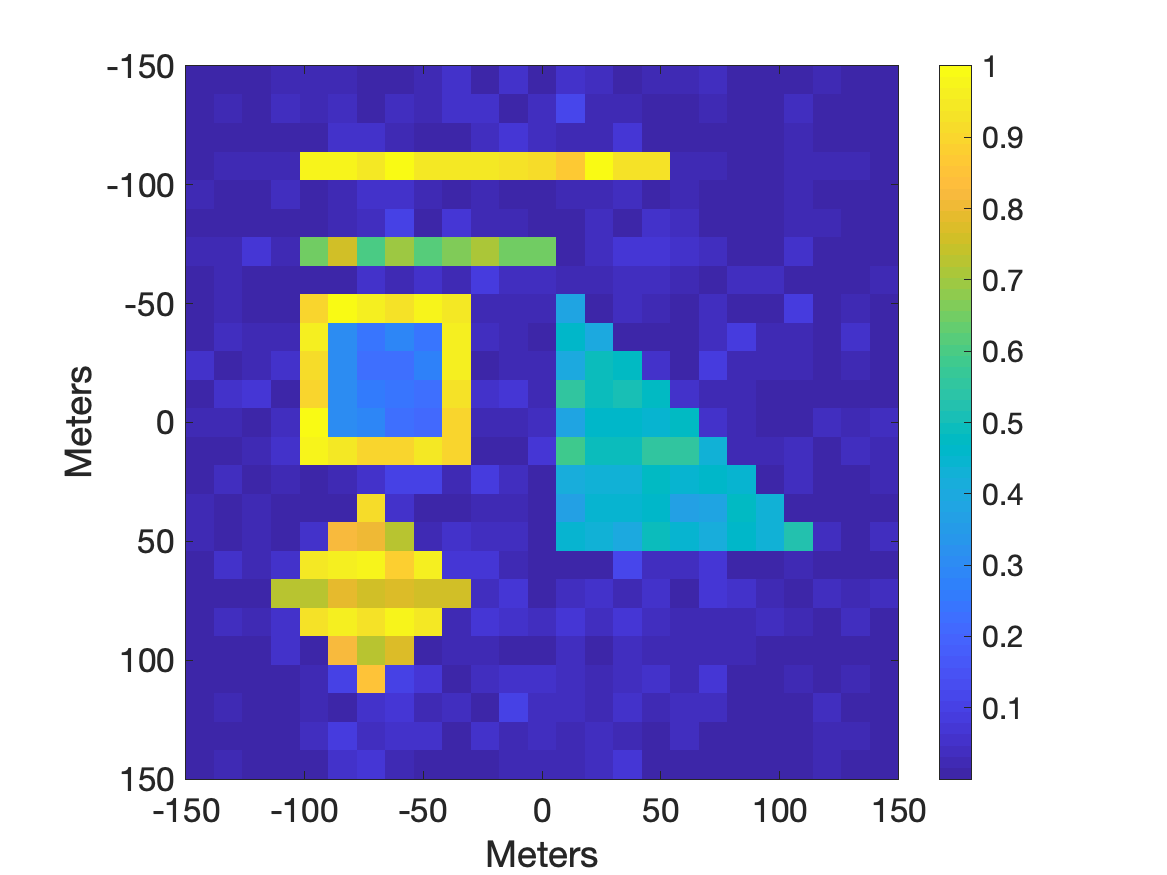}
    \caption{$24$ receivers.}\label{fig:24rec_passive}
  \end{subfigure}
  \caption{Sample reconstructions after $4000$ iterations of GWF for passive imaging case with varying number of receivers.
    Bandwidth was set at $10$ MHz with center frequency of $1.9$ GHz.   Number of frequency samples was held constant at $64$ and
    $K=625$. The pixel spacing was set at $12$ m. 
  }\label{fig:norec_recons_passive}
\end{figure}

~\figref{fig:norec} shows the MSE of the resulting reconstruction versus the number of
receivers for active and passive imaging.  Blue solid line is the result for the active case
while black dashed line is for the passive case.
For the active case, the bandwidth was held at $B=50$ MHz with the center frequency at
$\omega_c=10$ GHz for Fourier-based range resolution of $\Delta_{res}=3$ m.  For the passive case,
$B=10$ MHz and $\omega_c=1.9$ GHz for $\Delta_{res}=15$ m.  The pixel spacing was chosen such that it was smaller than
the Fourier-based range resolution for each case.  Namely, $\Delta = 2.4$ m and
$\Delta=12$ m for the active and passive cases, respectively.
The number of unknowns was held constant at
$K=625$ for both cases.  The GWF algorithm was performed for $4000$ iterations for comparison purposes.
Since the RIC directly
affects the rate of convergence of GWF, we expect to see smaller
MSE as the number of receivers grows.  This behavior is clearly present in both the active and passive
cases as can be readily observed in~\figref{fig:norec}.  In both cases, we observed exact convergence behavior
from $10$ receivers onward.  However, as expected, the convergence rate is slower with
smaller number of receivers.

As a visual confirmation of the experimental verification, sample reconstructions at two different
number of receivers ($12$ and $24$) is provided in~\textbf{Figures}
\textbf{\ref{fig:norec_recons_active}} and \textbf{\ref{fig:norec_recons_passive}}
for active and passive regimes,
respectively.

\subsection{Effect of Bandwidth/Range Resolution}\label{subsec:bw}

\begin{figure*}[]
  \centering
  \captionsetup[subfigure]{justification=centering}
  \begin{subfigure}[t]{0.33\textwidth}
    \centering
    \captionsetup[subfigure]{justification=centering}
    \includegraphics[width=\textwidth]{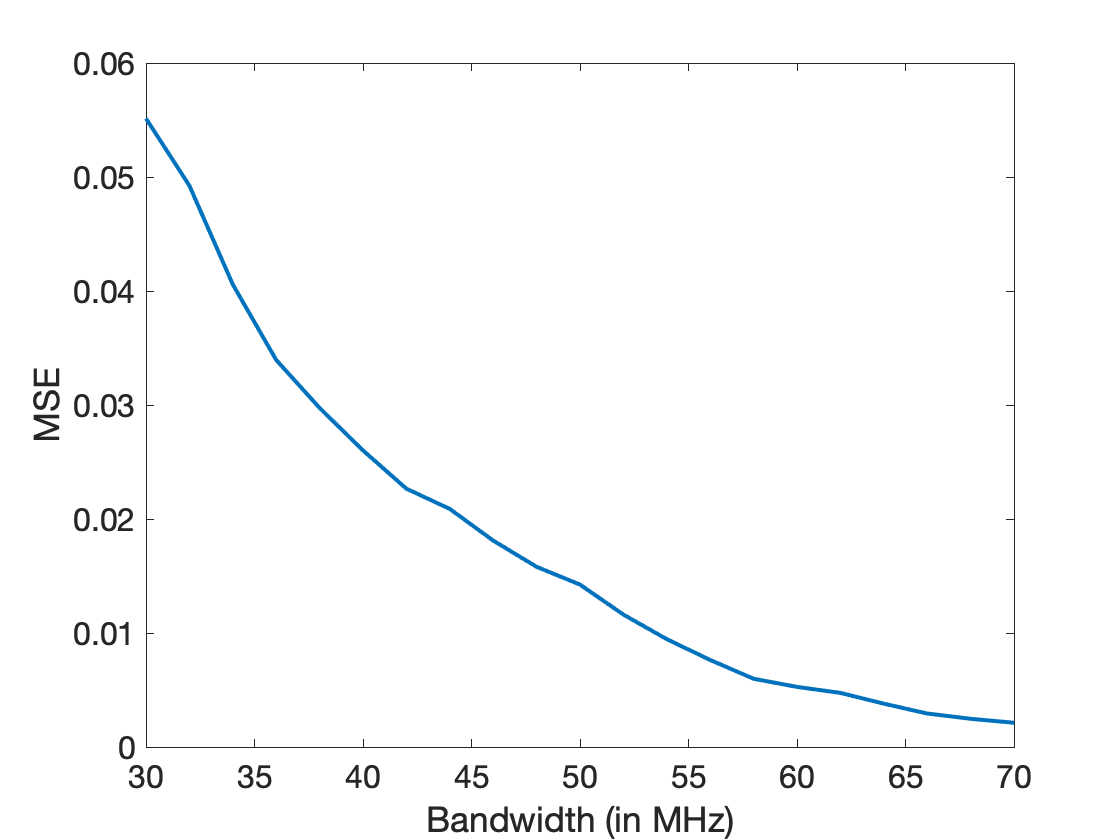}
    \caption{Active Regime. Center frequency was set at $10$ GHz and bandwidth ranged from $30$ MHz
      to $70$ MHz.}\label{fig:bw_active}
  \end{subfigure}\quad
  \begin{subfigure}[t]{0.33\textwidth}
    \centering
    \includegraphics[width=\textwidth]{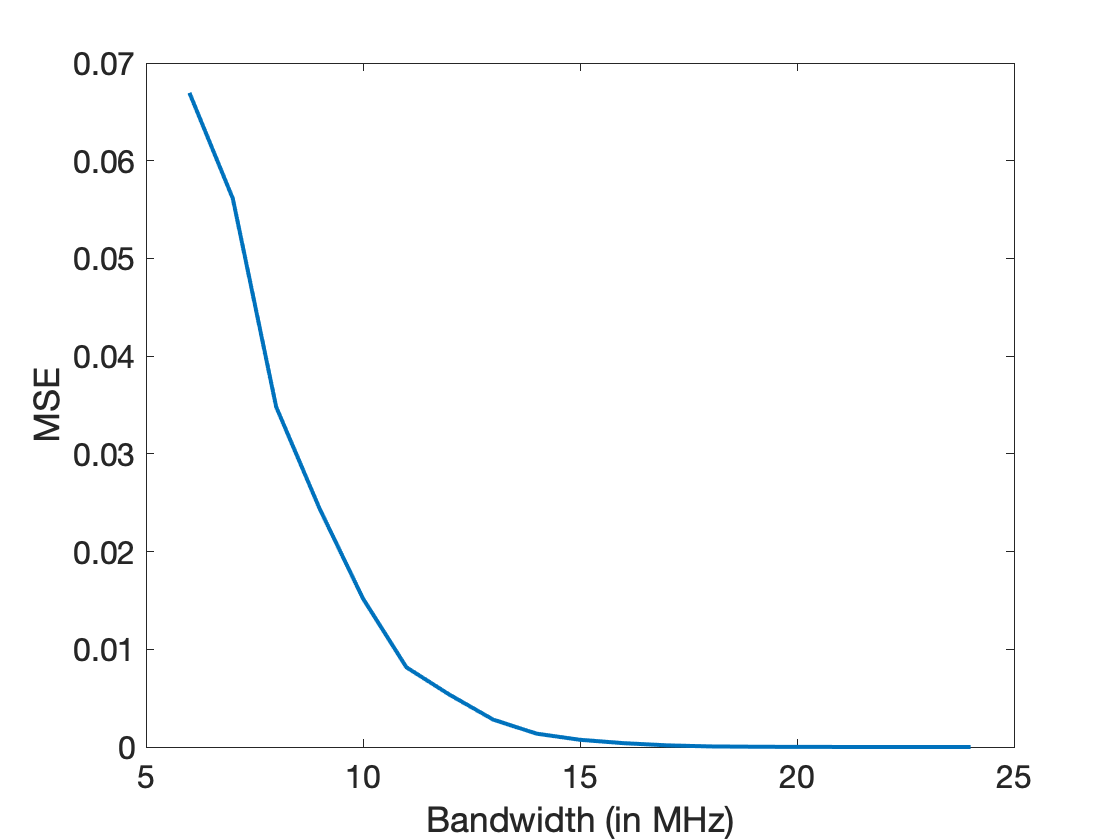}
    \caption{Passive Regime.  Center frequency was set at $1.9$ GHz and bandwidth ranged from $6$
      MHz to $24$ MHz.}\label{fig:bw_passive}
  \end{subfigure}
  \caption{Bandwidth vs. MSE of the reconstruction after $4000$ iterations of GWF for
    active and passive radar parameters.  Number of frequency samples was held constant at $64$ and
    $K=625$ for
    both cases. The pixel spacing was set at $2.4$m for active case and $12$m for passive.
  }\label{fig:bwmse}
\end{figure*}

\begin{figure}[!h]
  \centering
  \captionsetup[subfigure]{justification=centering}
  \begin{subfigure}[t]{0.23\textwidth}
    \centering
    \includegraphics[width=\textwidth]{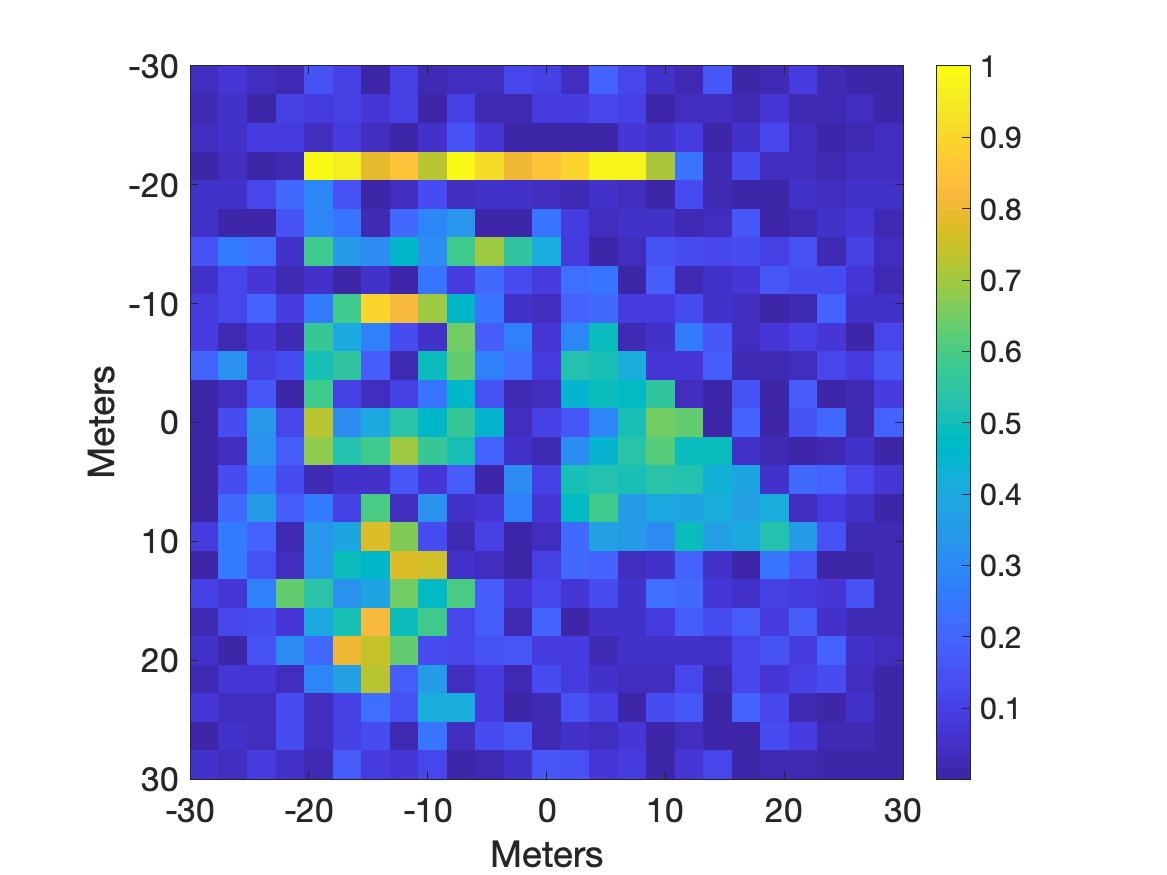}
    \caption{$40$ MHz bandwidth.}\label{fig:B40_active}
  \end{subfigure}\quad
  \begin{subfigure}[t]{0.23\textwidth}
    \centering
    \includegraphics[width=\textwidth]{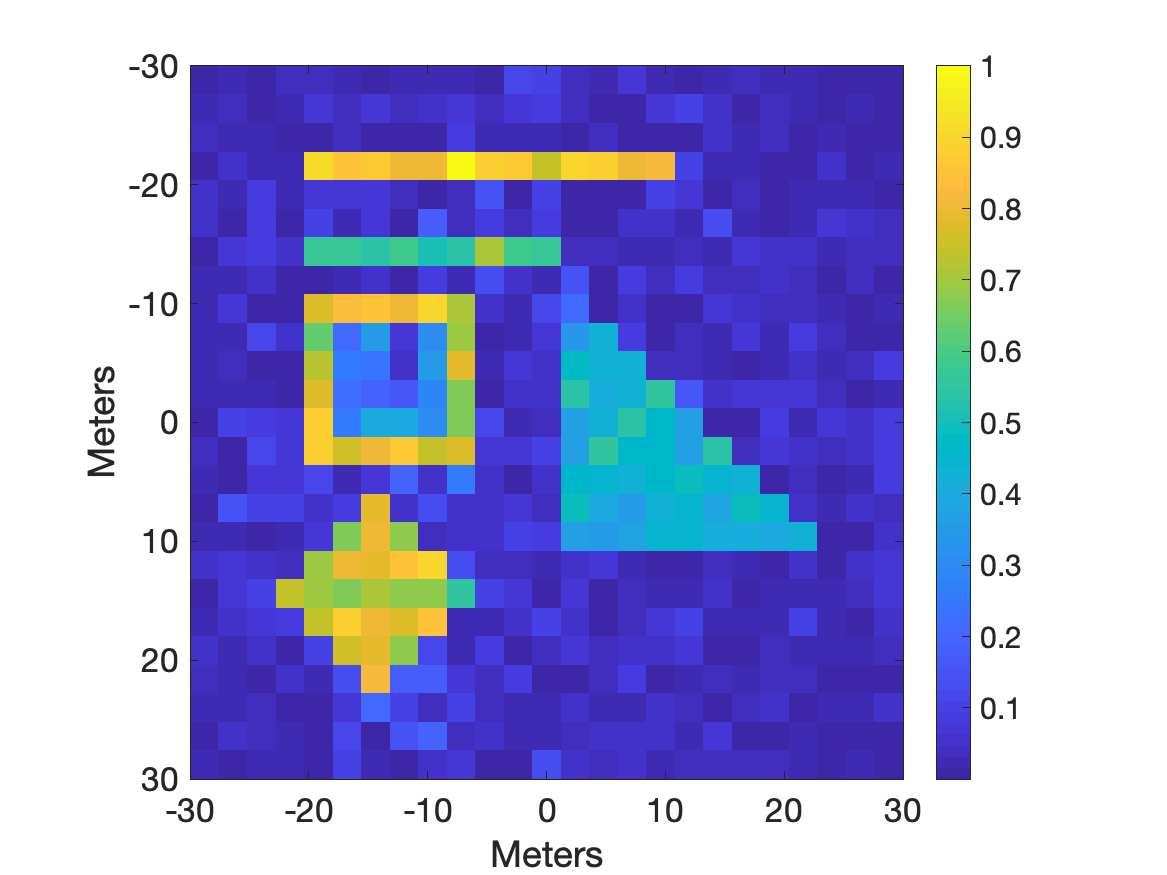}
    \caption{$60$ MHz bandwidth.}\label{fig:B60_active}
  \end{subfigure}
  \caption{Sample reconstructions after $4000$ iterations of GWF for active imaging case with varying bandwidth.
    $18$ receivers were used for reconstruction with center frequency of $10$ GHz.   Number of
    frequency samples was held constant at $64$ and
    $K=625$. 
    The pixel spacing was set at $2.4$ m.
  }\label{fig:B_recons_active}
\end{figure}
\begin{figure}[!h]
  \centering
  \captionsetup[subfigure]{justification=centering}
  \begin{subfigure}[t]{0.23\textwidth}
    \centering
    \includegraphics[width=\textwidth]{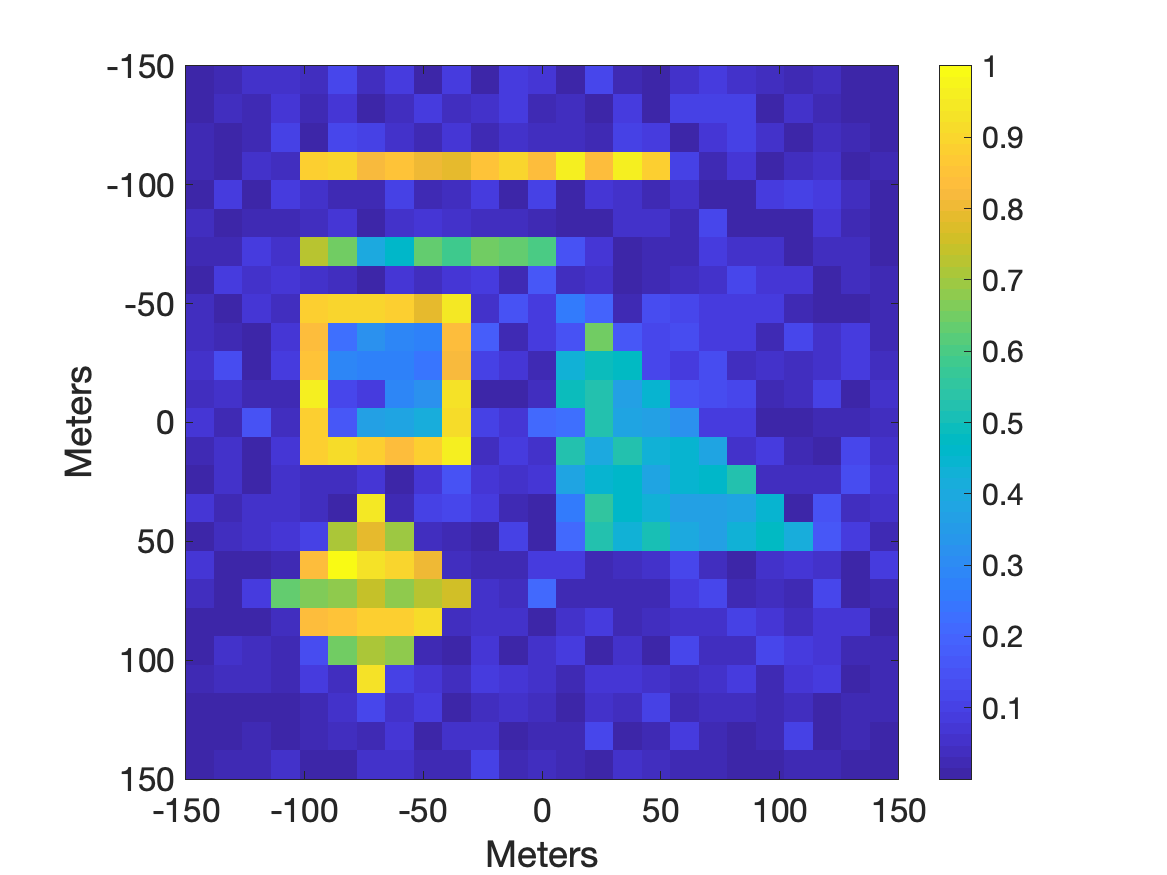}
    \caption{$12$ MHz bandwidth.}\label{fig:B12_passive}
  \end{subfigure}\quad
  \begin{subfigure}[t]{0.23\textwidth}
    \centering
    \includegraphics[width=\textwidth]{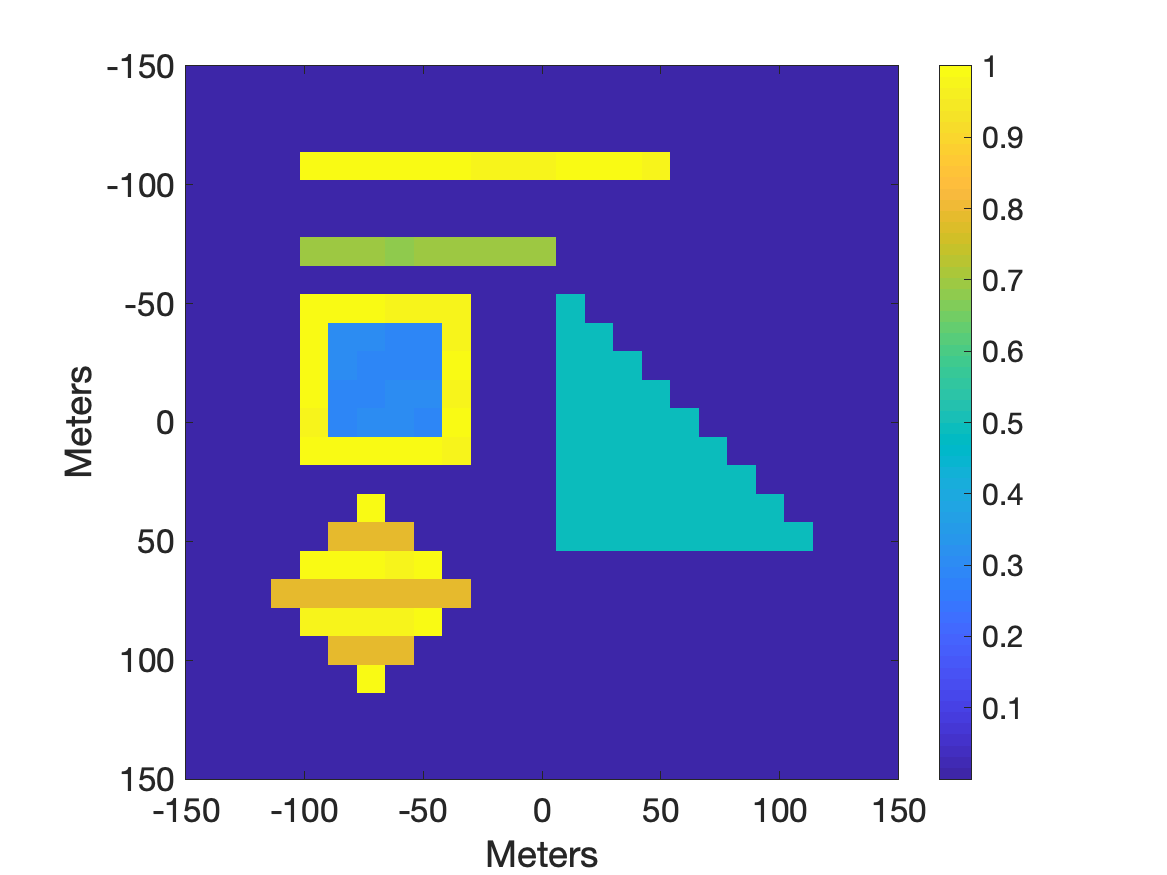}
    \caption{$20$ MHz bandwidth.}\label{fig:B20_passive}
  \end{subfigure}
  \caption{Sample reconstructions after $4000$ iterations of GWF for passive imaging case with varying bandwidth.
    $18$ receivers were used for reconstruction with center frequency of $1.9$ GHz.   Number of
    frequency samples was held constant at $64$ and
    $K=625$. 
    The pixel spacing was set at $12$ m.
  }\label{fig:B_recons_passive}
\end{figure}

Next we examine the effect of the bandwidth on the convergence behavior.
Both terms in~\eqref{eq:delta} includes square root of $\Delta_{res}$, the range resolution, in the
numerator.
This suggests that there is an inverse relationship between 
the bandwidth and RIC.
Similar to above, we test the effect of bandwidth on the convergence behavior of GWF algorithm, and
we ran a series of GWF reconstruction on the same scene while varying the bandwidth and
holding other relevant parameters fixed.  The number of receivers used for the experiments was
fixed at $N=18$.
All other parameters were held to the same values as in the previous subsection.



\figref{fig:bwmse} summarizes the result of these experiments.~\figref{fig:bw_active} shows the
bandwidth vs. MSE curve for active case.  We varied the bandwidth in the range of $30$ MHz to $70$
MHz.~\figref{fig:bw_passive} shows the same curve for the passive case where the bandwidth was
varied between $6$ MHz and $24$ MHz.
Examining the two figures, we clearly see that higher bandwidth
results in smaller MSE, and hence faster convergence to exact solution.
This agrees with the theoretical bound in~\eqref{eq:delta}.
As before, we provide visual
confirmation in form of sample reconstructions in~\textbf{Figures}
\textbf{\ref{fig:B_recons_active}} and \textbf{\ref{fig:B_recons_passive}}
for active and passive regimes,
respectively.

\subsection{Effect of Center Frequency}\label{subsec:wc}

\begin{figure*}[]
  \centering
  \captionsetup[subfigure]{justification=centering}
  \begin{subfigure}[t]{0.33\textwidth}
    \centering
    \captionsetup[subfigure]{justification=centering}
    \includegraphics[width=\textwidth]{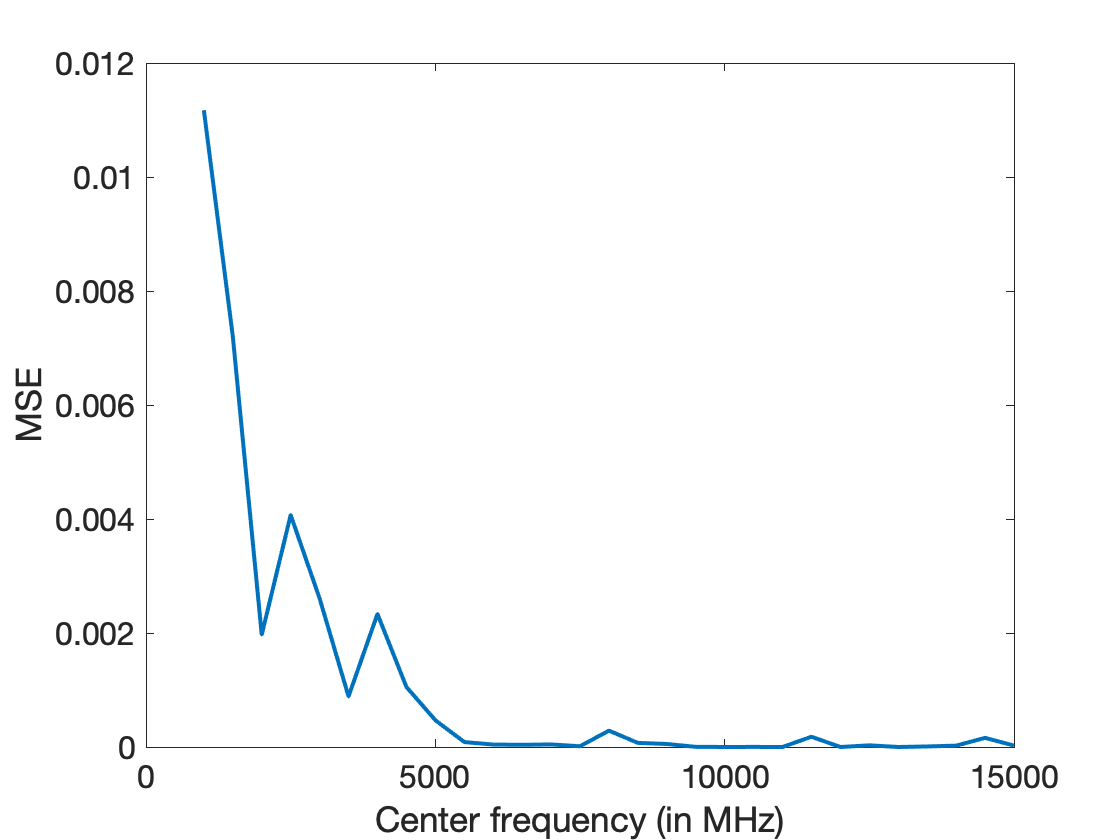}
    \caption{Active Regime. }\label{fig:wc_active}
  \end{subfigure}\quad
  \begin{subfigure}[t]{0.33\textwidth}
    \centering
    \includegraphics[width=\textwidth]{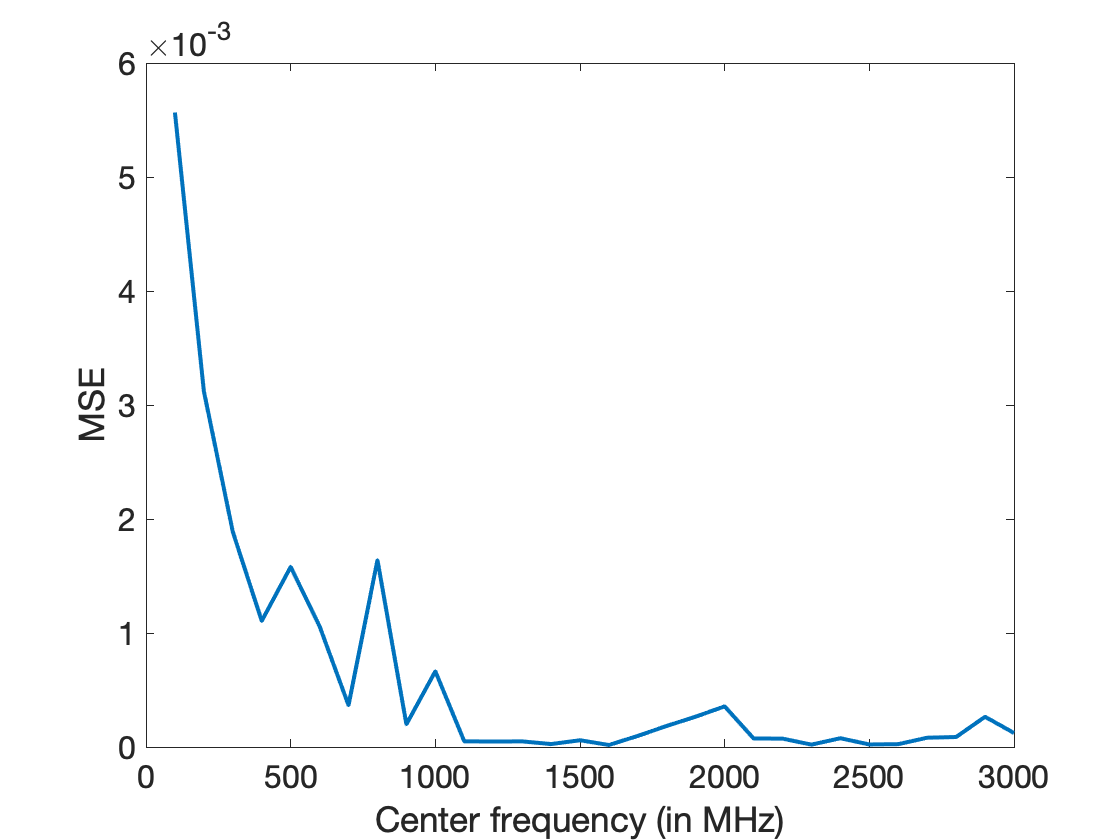}
    \caption{Passive Regime. }\label{fig:wc_passive}
  \end{subfigure}
  \caption{Center frequency vs. MSE of the reconstruction after $4000$ iterations of GWF for
    active and passive radar parameters.  Number of frequency samples was held constant at $64$ and
    $K=625$ for
    both cases. The pixel spacing was set at $2.4$m for active case and $12$m for passive.  The
    bandwidth was fixed at $50$ MHz and $10$ MHz for active and passive cases, respectively. 
  }\label{fig:wcvsmse}
\end{figure*}

\begin{figure}[!ht]
  \centering
  \captionsetup[subfigure]{justification=centering}
  \begin{subfigure}[t]{0.23\textwidth}
    \centering
    \includegraphics[width=\textwidth]{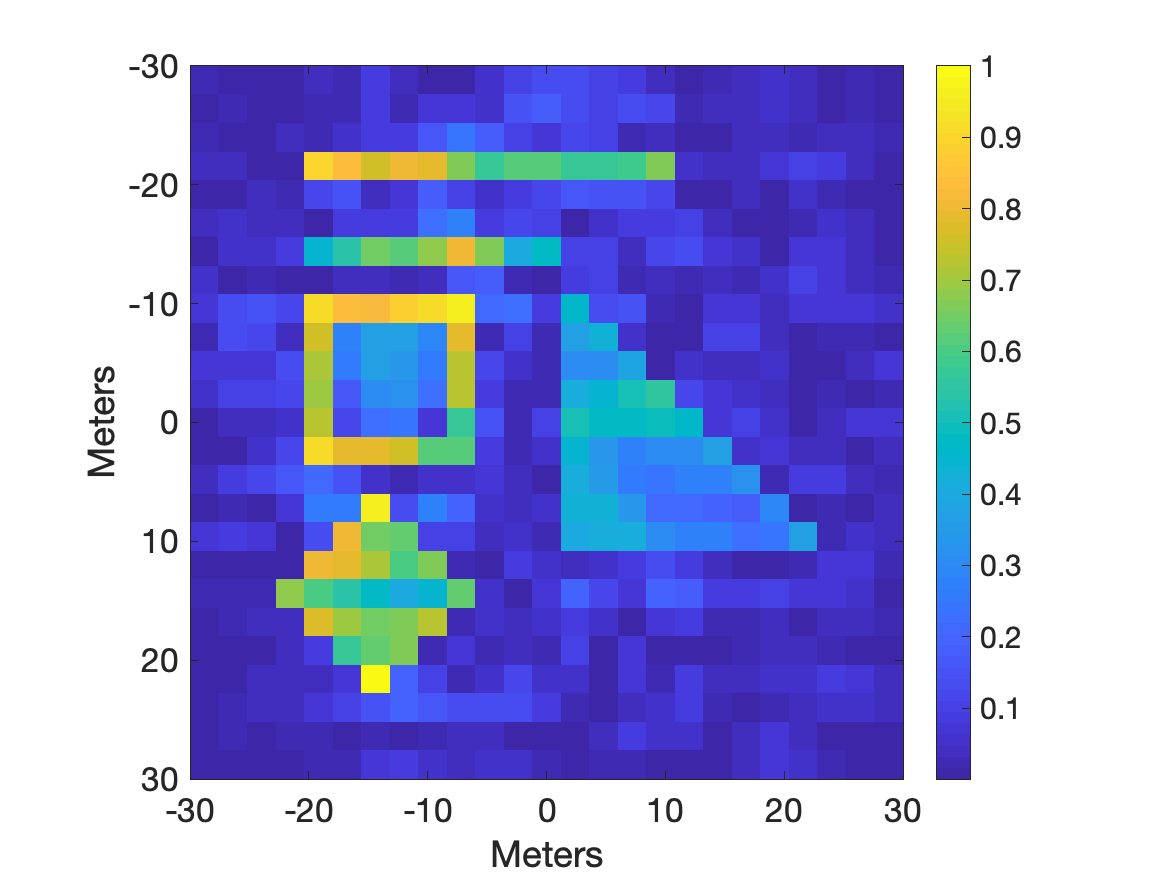}
    \caption{$1$ GHz center frequency.}\label{fig:fc5000_active}
  \end{subfigure}\quad
  \begin{subfigure}[t]{0.23\textwidth}
    \centering
    \includegraphics[width=\textwidth]{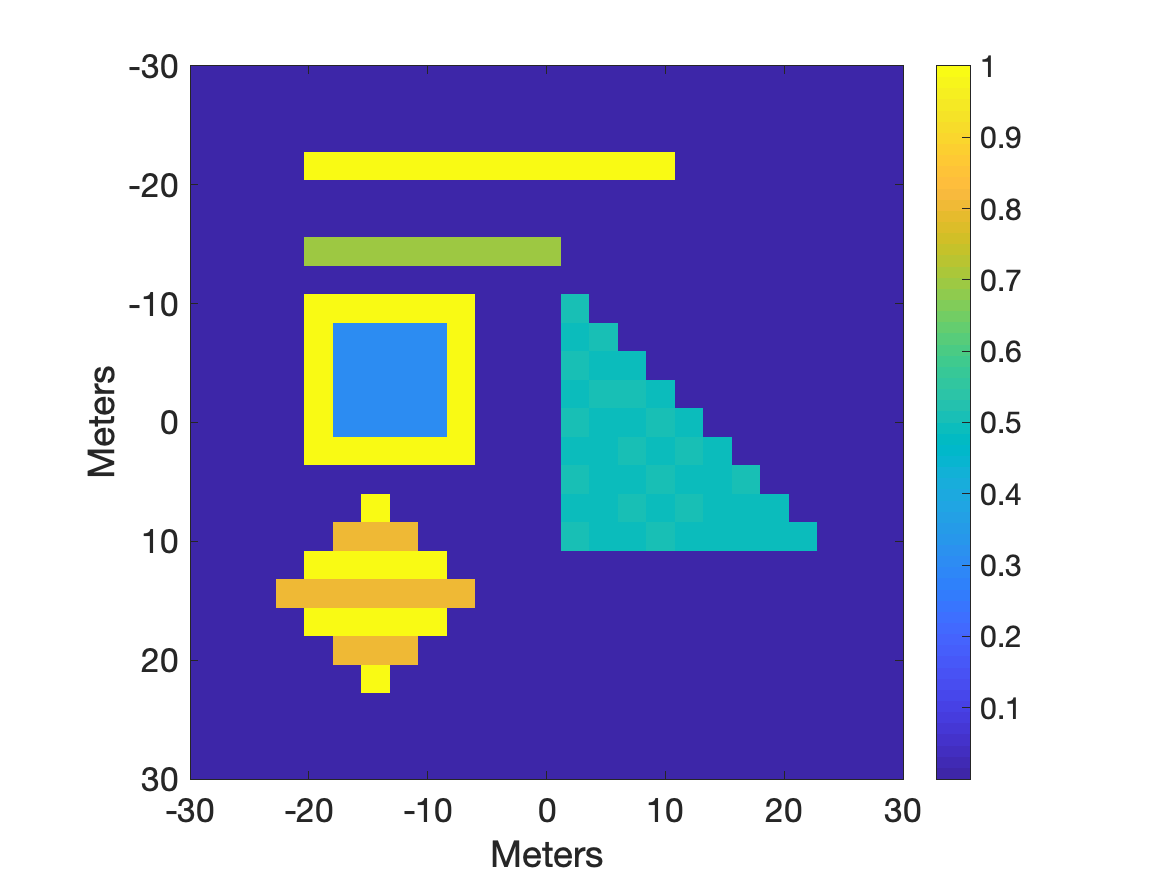}
    \caption{$12$ GHz center frequency.}\label{fig:fc12000_active}
  \end{subfigure}
  \caption{Sample reconstructions after $4000$ iterations of GWF for active imaging case with
    varying center frequencies.
    $32$ receivers were used for reconstruction with bandwidth fixed at $50$ MHz.   Number of
    frequency samples were set as $64$ and
    $K=625$. 
    The pixel spacing was set at $2.4$ m.
  }\label{fig:wc_recons_active}
\end{figure}
\begin{figure}[!ht]
  \centering
  \captionsetup[subfigure]{justification=centering}
  \begin{subfigure}[t]{0.23\textwidth}
    \centering
    \includegraphics[width=\textwidth]{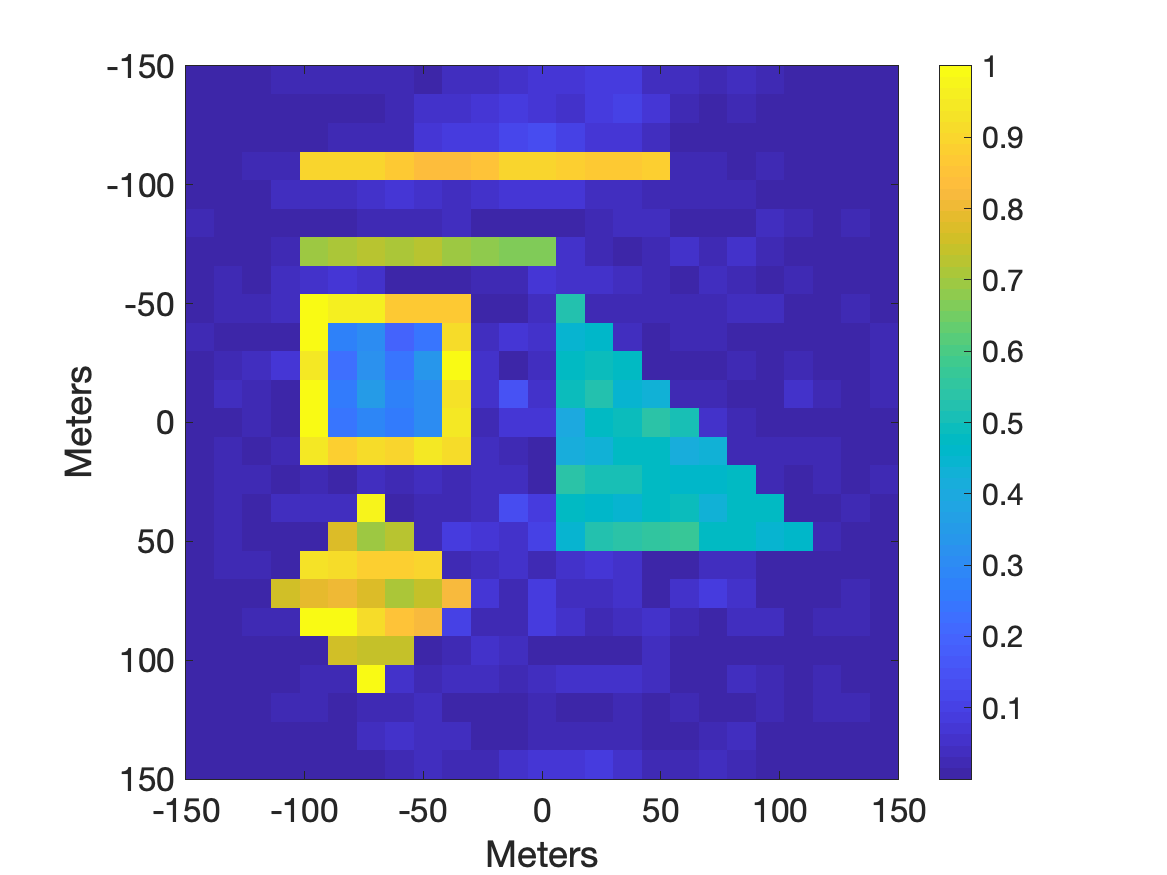}
    \caption{$1$ GHz center frequency.}\label{fig:./passive_imgs/fc1000-16rec}
  \end{subfigure}\quad
  \begin{subfigure}[t]{0.23\textwidth}
    \centering
    \includegraphics[width=\textwidth]{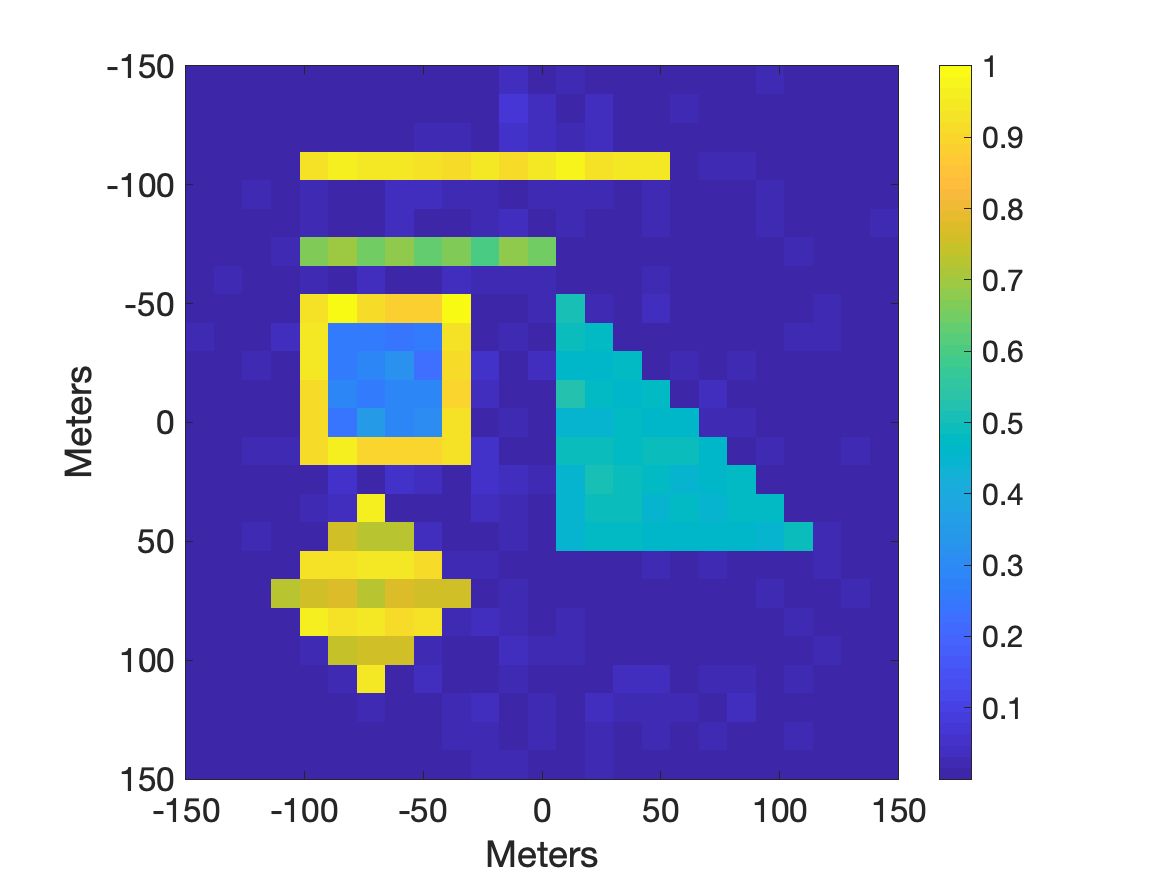}
    \caption{$2$ GHz center frequency.}\label{fig:fc2000_passive}
  \end{subfigure}
  \caption{Sample reconstructions after $4000$ iterations of GWF for passive imaging case with
    varying center frequencies.
    $32$ receivers were used for reconstruction with bandwidth fixed at $10$ MHz.   Number of
    frequency samples were set as $64$ and
    $K=625$. 
    The pixel spacing was set at $12$ m.
  }\label{fig:wc_recons_passive}
\end{figure}

The first term of~\eqref{eq:delta} is 
inversely proportional to the center frequency of the transmitted
waveform and as such we expect the center frequency to improve the convergence behavior of GWF as
the center frequency gets larger.
We examined numerically, the effect of center frequency on the exact
reconstruction and the convergence rate by, again, running a series of numerical simulations where
we varied the center frequency while keeping other
relevant variables constant.  To minimize the effect of the second term on the
RIC and better evaluate the impact of central frequency in the super-resolution regime, we increased the number of receivers used in these experiments to $N=32$ for both cases.

\textbf{Figures} \textbf{\ref{fig:wc_active}} and \textbf{\ref{fig:wc_passive}} show the results of
simulated experiments for active and passive scenarios, respectively.  For active case, we varied
the center frequency in the range between $0.5$ GHz and $15$ GHz.  For the passive case, the range
was restricted to $0.1$ GHz to $3$ GHz to reflect realistic values for sources of opportunity.
In both cases, we observe a behavior of downward trend in MSE as the center frequency
increases, albeit, not as drastic as in other parameters. 
This is due to the fact that the central frequency appears inverted in the two terms of the RIC upper bound. 
The decaying trend of our experiments agrees with the notion that the order constant of the second term in the RIC upper bound adequately suppresses $\lambda_c^{-3/2}$, as $N^2 = \mathcal{O}(K)$ proves to be sufficient for super-resolution. 

Notice, however, that in the active case, larger center frequency value is needed to
achieve similar performance as in the passive case.  This is attributable to the fact that the first
term is proportional to $\sqrt{L\Delta_{res}}/\Delta^2$.  With the active parameters, this term is
approximately $8$ times that of the passive case.  Thus, the center frequency needs to be higher to
compensate for the difference.~\textbf{Figures}
\textbf{\ref{fig:wc_recons_active}} and \textbf{\ref{fig:wc_recons_passive}} show  sample
reconstructions at two different center frequencies for active and passive regimes, respectively.

\subsection{Effect of Additive Noise}\label{subsec:noise}

\begin{figure}
  \centering
  \includegraphics[width=0.33\textwidth]{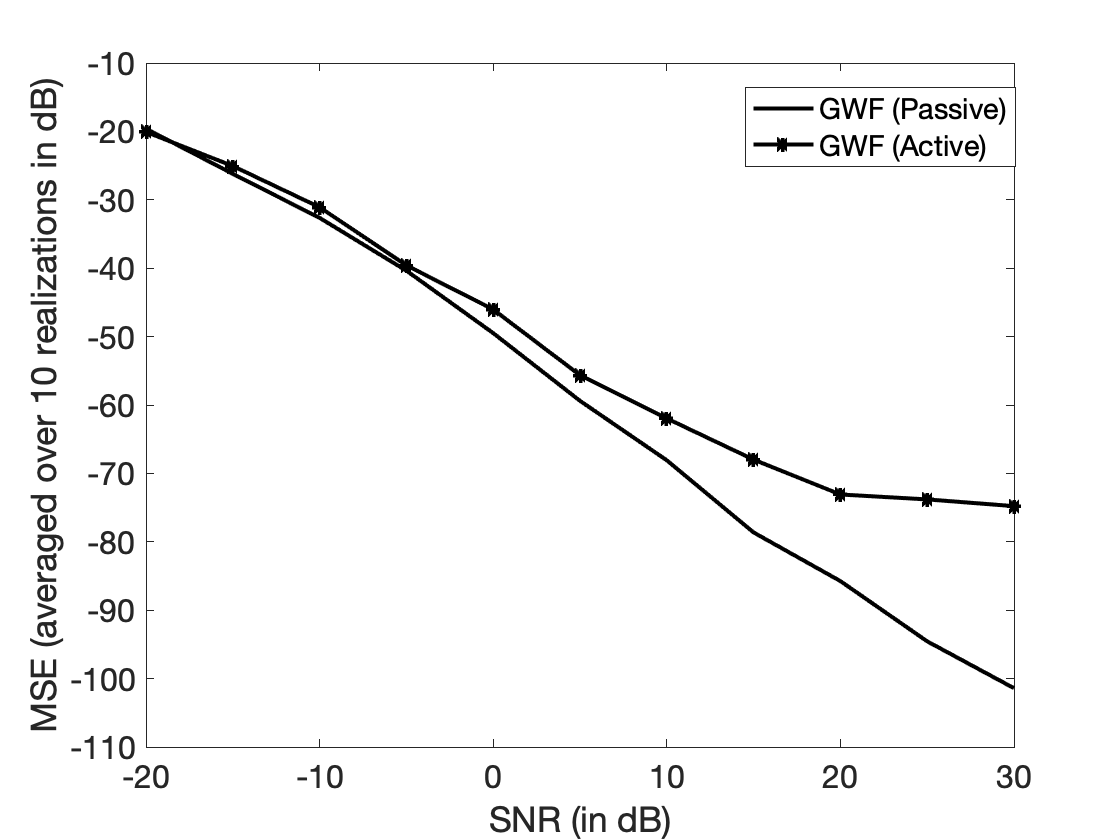}
  \caption{Received signal SNR vs. MSE of the reconstruction after $4000$ iterations of GWF. SNR values depict those in the linear measurements collected at each receiver, prior to correlations. 
    Dotted and solid curves are for active and passive radar parameters, respectively.  Number of frequency samples were fixed as $64$, with
    $K=625$ for
    both cases. The pixel spacing was set at $2.4$m for active case and $12$m for passive.  The
    center frequency was set at $10$ GHz and $1.9$ GHz for active and passive cases, respectively.
    The bandwidth was set at $50$ MHz and $10$ MHz for active and passive cases, respectively.
  }\label{fig:nz}
\end{figure}

\begin{figure}[!h]
  \centering
  \captionsetup[subfigure]{justification=centering}
  \begin{subfigure}[t]{0.23\textwidth}
    \centering
    \includegraphics[width=\textwidth]{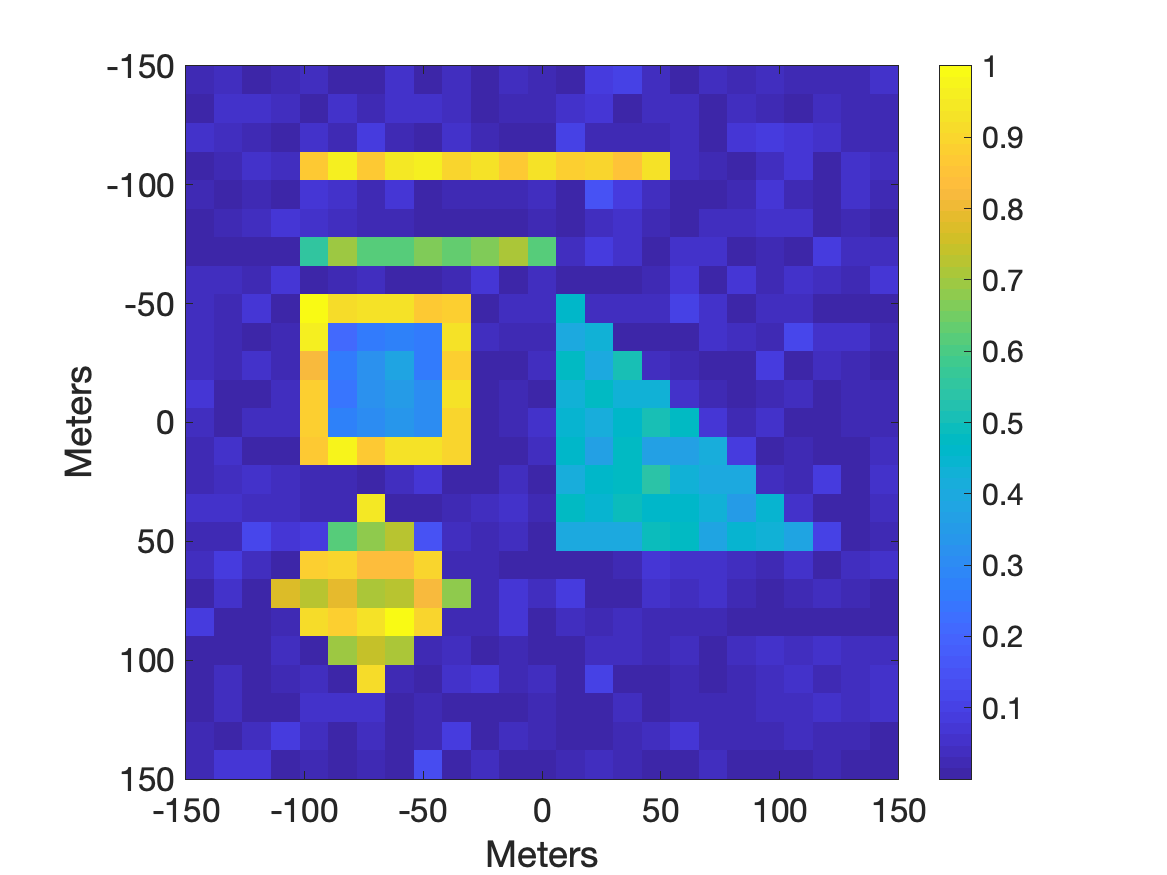}
    \caption{$0$ dB SNR (passive).}\label{fig:nz-0dB}
  \end{subfigure}\quad
  \begin{subfigure}[t]{0.23\textwidth}
    \centering
    \includegraphics[width=\textwidth]{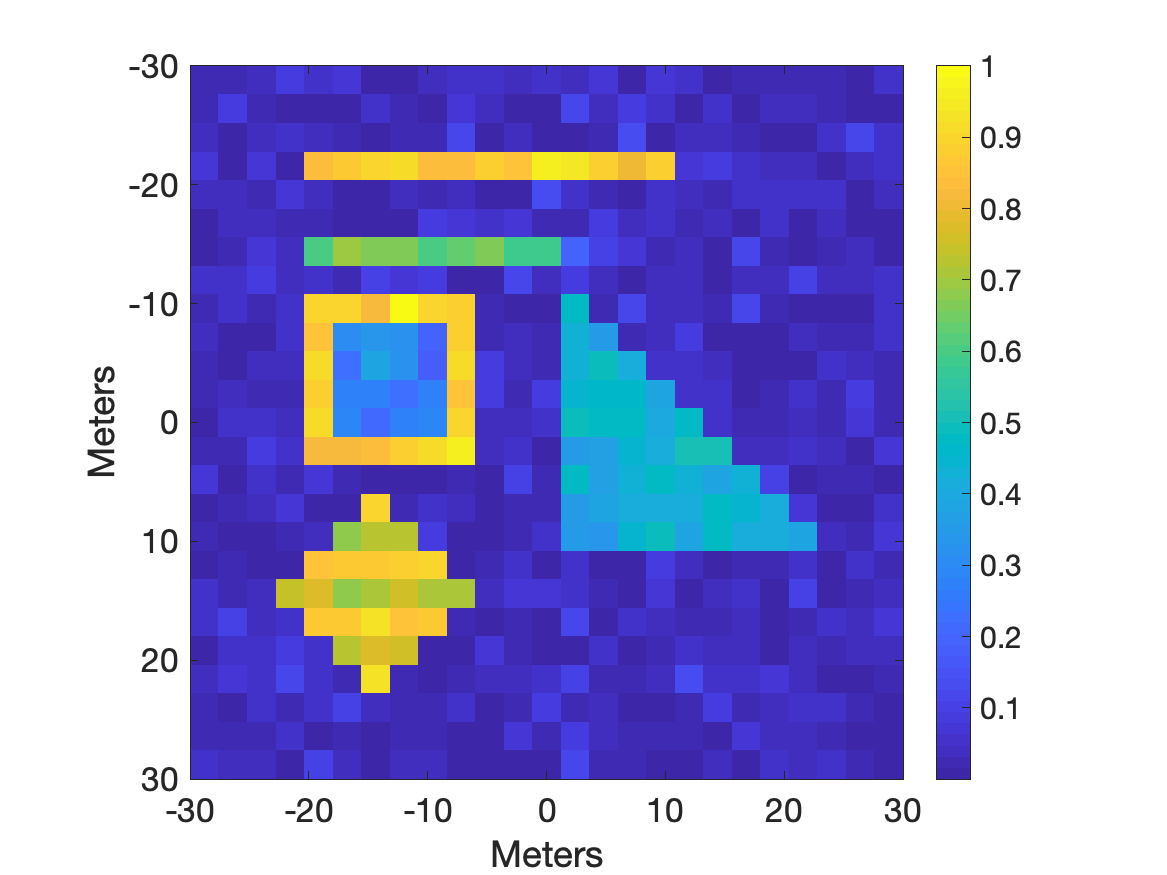}
    \caption{$0$ dB SNR (active).}\label{fig:nz-0dB_active}
  \end{subfigure}
  \caption{Sample reconstructions after $4000$ iterations of GWF with $0$ dB 
    SNR at the receivers.
    $30$ receivers were used for reconstruction. Number of
    frequency samples was held constant at $64$ and
    $K=625$. 
    The pixel spacing was set at $12$ m, and $2.4$ m, respectively. 
  }\label{fig:nz_recon0}
\end{figure}
\begin{figure}[!h]
  \centering
  \captionsetup[subfigure]{justification=centering}
  \begin{subfigure}[t]{0.23\textwidth}
    \centering
    \includegraphics[width=\textwidth]{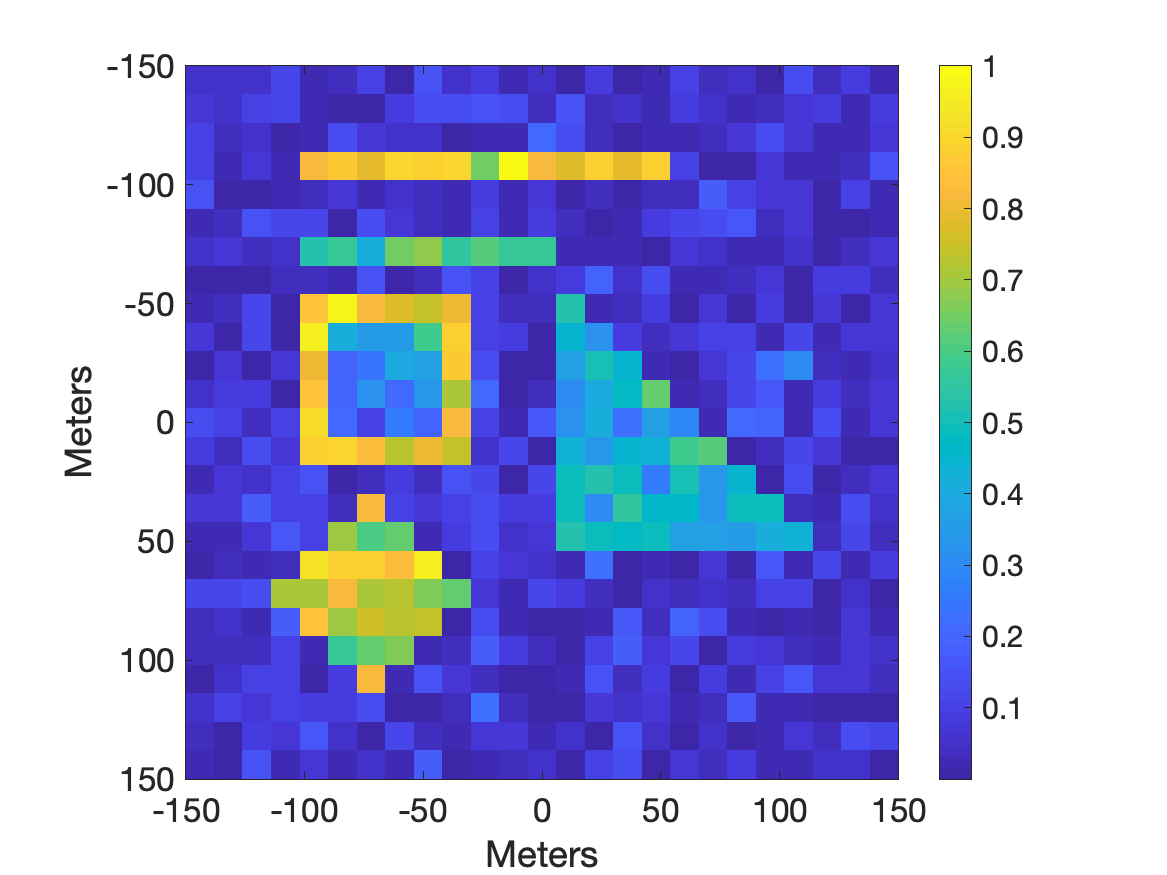}
    \caption{$-5$ dB SNR (passive).}\label{fig:nz--5dB}
  \end{subfigure}\quad
  \begin{subfigure}[t]{0.23\textwidth}
    \centering
    \includegraphics[width=\textwidth]{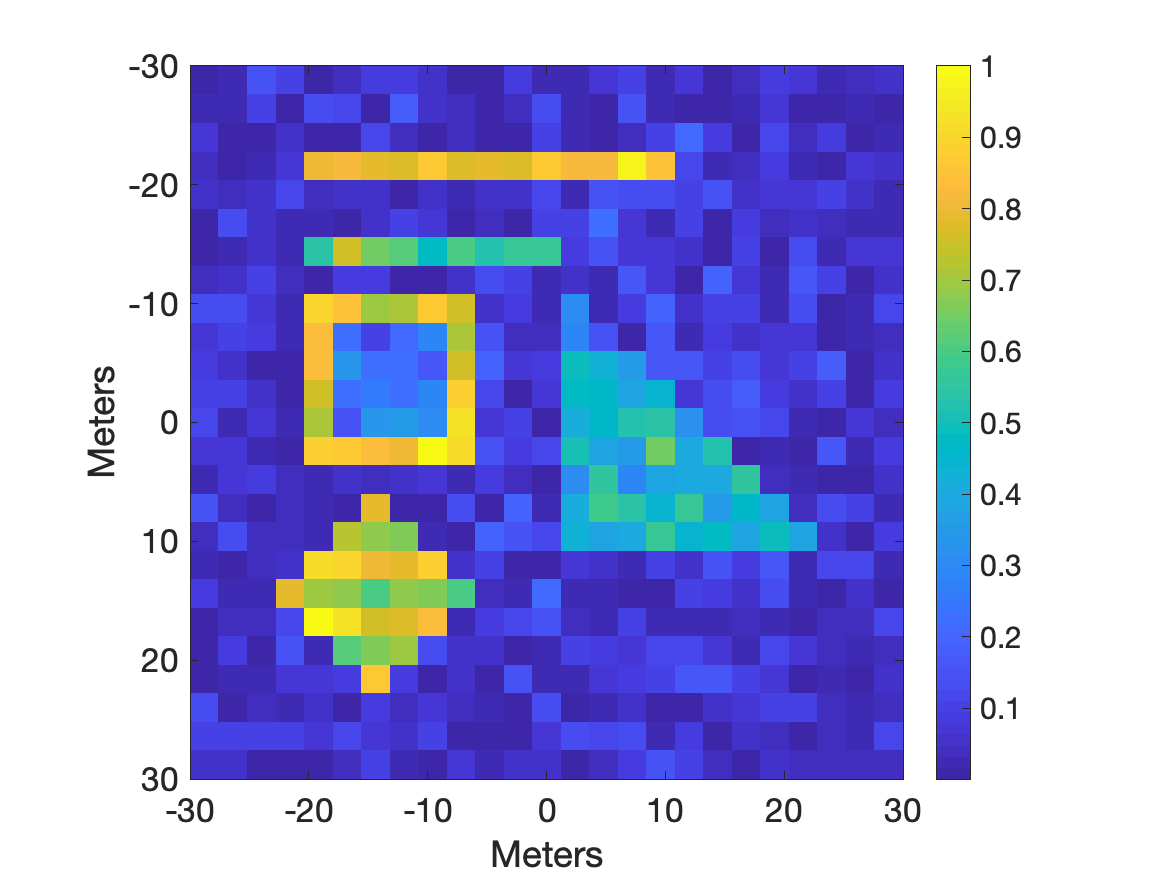}
    \caption{$-5$ dB SNR (active).}\label{fig:nz-5dB_active}
  \end{subfigure}
  \caption{Sample reconstructions after $4000$ iterations of GWF with $-5$ dB 
    SNR at the receivers.
    $30$ receivers were used for reconstruction. Number of
    frequency samples was held constant at $64$ and
    $K=625$. 
    The pixel spacing was set at $12$ m, and $2.4$ m, respectively. 
  }\label{fig:nz_recon-5}
\end{figure}
\begin{figure}[!h]
  \centering
  \captionsetup[subfigure]{justification=centering}
  \begin{subfigure}[t]{0.23\textwidth}
    \centering
    \includegraphics[width=\textwidth]{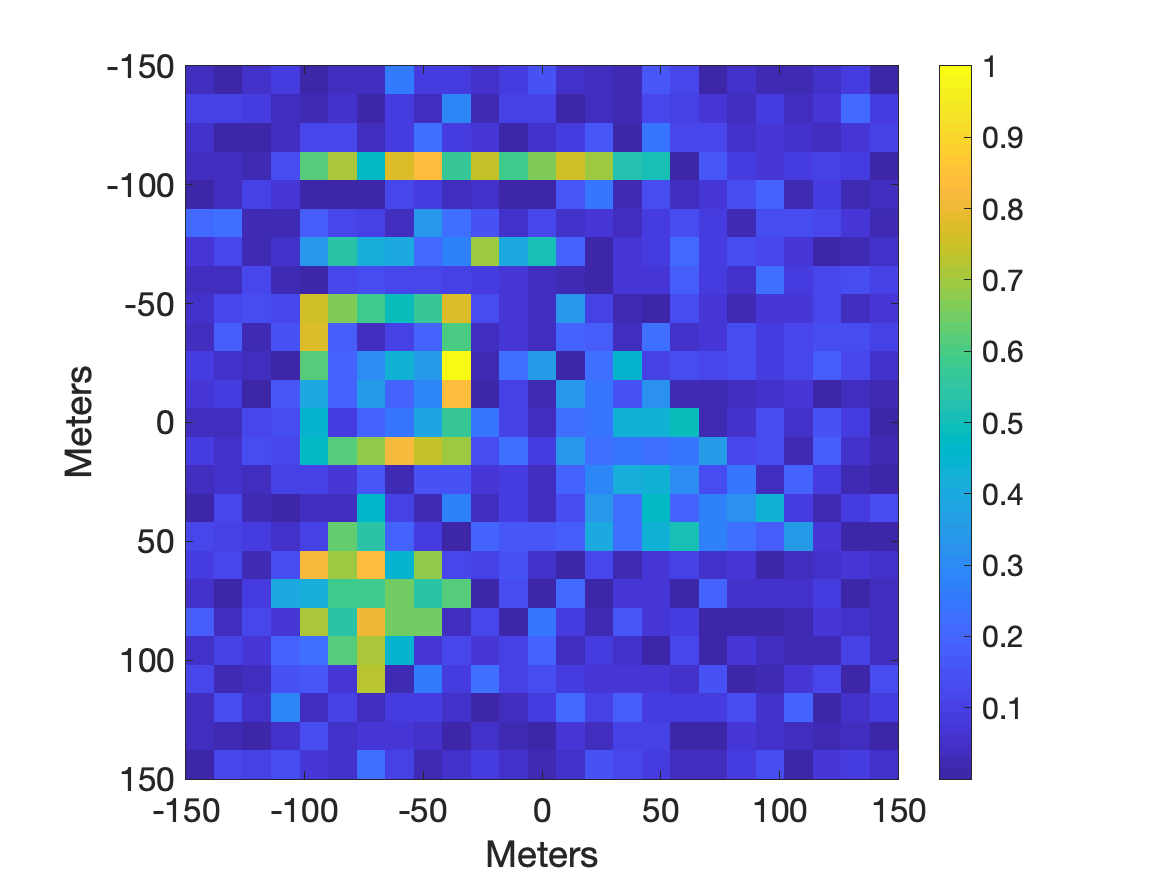}
    \caption{$-10$ dB SNR (passive).}\label{fig:nz-10dB}
  \end{subfigure}\quad
  \begin{subfigure}[t]{0.23\textwidth}
    \centering
    \includegraphics[width=\textwidth]{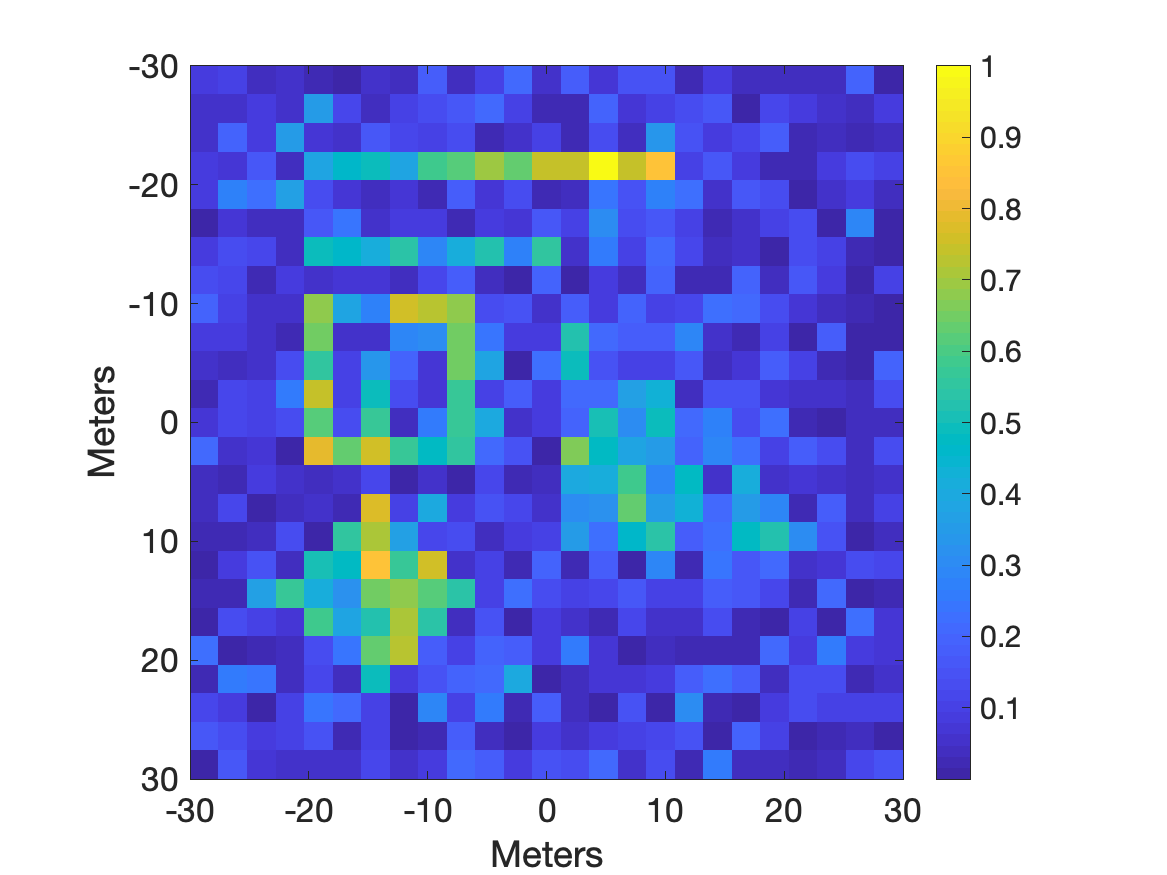}
    \caption{$-10$ dB SNR (active).}\label{fig:nz-10dB_active}
  \end{subfigure}
  \caption{Sample reconstructions after $4000$ iterations of GWF with $-10$ dB 
    SNR at the receivers.
    $30$ receivers were used for reconstruction. Number of
    frequency samples was held constant at $64$ and
    $K=625$. 
    The pixel spacing was set at $12$ m, and $2.4$ m, respectively. 
  }\label{fig:nz_recon-10}
\end{figure}

Next, we evaluate the robustness of the proposed method for interferometric multi-static radar imaging, in passive, and active cases. 
As before, we use a flat spectrum signal with $10$ MHz bandwidth centered around $1.9$ GHz frequency in the passive, and a $50$ MHz bandwidth centered around $10$ GHz frequency in the active experiments. 
We incorporate additive, zero-mean white Gaussian noise on the linear signal model at the receivers, and consider SNR levels varying from $-20$ to $30$ dB. 
~\figref{fig:nz} demonstrates the reconstruction MSE with respect to the received signal SNR with $5$ dB increments, with errors averaged over $10$ realizations. 
The MSE curves indicate that GWF is robust to additive noise at the receivers in both active and passive cases, and the performance of the algorithm degrades predictably with decreasing SNRs, with steady decay in MSE as the SNR at the receivers improve. 

It should be noted that the correlation operation amplifies the noise variance, hence the interferometric measurements processed in the experiments have lower SNR than the specified levels at the receivers. 
Nonetheless, we observe that the reconstruction performance of GWF degrades gracefully as SNR decrases, as ~\figref{fig:nz_recon0}, ~\figref{fig:nz_recon-5} and ~\figref{fig:nz_recon-10} demonstrate that GWF is capable of producing highly accurate imagery with $0$ dB in the received signals at the antennas. 
However with SNRs below $-10$ dB the algorithm performance degrades noticeably, which motivates filtering or sample truncation at the implementation of GWF in low SNR scenarios.

\section{Conclusion}\label{sec:conc}

In this paper, we utilize GWF theory developed in~\cite{bariscan2018} for exact multistatic imaging of extended targets by designing the underlying imaging parameters such that the sufficient condition for exact recovery is satisfied.
Our work has two significant contributions.
1) Unlike the state-of-the-art interferometric inversion methods based on LRMR, GWF avoids lifting the
problem. As a result, it is computationally efficient and does not require large memory allocations, making it
suitable for practical applications.
2) We demonstrate that the underlying imaging parameters can be designed so that the RIP over rank-1, PSD matrices is satisfied by a \emph{deterministic} lifted forward model.

We first show the asymptotic isometry of the lifted forward model, $\mathcal{F}$, of interferometric multistatic radar, as the center frequency and the number of receivers go to
infinity.
We then proceed with estimating the deviation from the asymptotic behavior when imaging parameters are finite, and derive an upper bound for the RIC of $\mathcal{F}$ over the set of rank-1, PSD matrices.
Hence, we identify the relation of imaging parameters to the sufficient condition of exact recovery. 
Using the RIC upper bound, we determine a lower limit for pixel spacing to achieve exact recovery.
This limit is superior to the Fourier-based range resolution 
for sufficiently small scenes.
Furthermore, we determine the minimal sample complexity needed for RIC upper bound to be sufficiently small, hence identify the practical requirements for reconstruction when designing a multistatic imaging system.
In our numerical simulations, we evaluate the impact of the imaging parameters in our upper bound estimate of
RIC 
in reconstruction performance, verify our theoretical results, and finally assess the performance of GWF to additive noise at low SNRs. 

For future work, we will study the robustness of our method with respect to deviations from our imaging setup, such as non-equi-distant locations, or non-circular configurations of receivers. 
In addition, we will investigate extensions of our theory 
to the case involving additive noise and outliers, moving target imaging, as well as implementation of our method using real scattering data. 

\bibliographystyle{IEEEtran}
\bibliography{citations_BY}
\appendix
\section{Proofs}
\subsection{Proof of Lemma~\ref{lem:lem1}}\label{prf:lem1}
We first examine the $2$-norm of the data.
For a rank-$1$ $\tilde{\brho}$, we have that $\norm{\tilde{\brho}}_F^2 = \norm{\brho}_2^4$.  We can
also rewrite
\begin{equation}
  \norm{\mathcal{F}\tilde{\brho}}^2_2 = \frac{1}{M\binom{N}{2}}\sum_{i=1}^N\sum_{i<j}^N\sum_{m=1}^M \abs{\ip{\L_i^m}{\brho}}^2\abs{\ip{\L_j^m}{\brho}}^2.\label{eq:norm2}
\end{equation}
Thus, from~\eqref{eq:Li},~\eqref{eq:varphi2}, and~\eqref{eq:Aij2} we have
\begin{equation}
  \begin{aligned}
  \abs{\ip{\L_i^m}{\brho}}^2 =
  \sum_{k,k'}&\mrm{e}^{-\mrm{i}\omega_m/c_0(\ip{\hat{\a}_i^r}{\x_k-\x_{k'}}+\ip{\hat{\a}^t}{\x_k-\x_{k'}})}\\
  &\times\rho(\bm{x}_k)\rho^{*}(\bm{x}_{k'})\frac{\abs{C_i}^2}{\abs{\a_i^r}^2\abs{\a^t}^2}.
  \end{aligned}
  \label{eq:ip2}
\end{equation}
Similarly, we have that
\begin{equation}
  \begin{aligned}
  &\abs{\ip{\L_i^m}{\brho}}^2\abs{\ip{\L_j^m}{\brho}}^2=
  \sum_{k,k',l,l'}\mrm{e}^{-\mrm{i}\omega_m/c_0\Phi_{i,j}^{k,k',l,l'}}\\
  &\qquad\times\rho(\bm{x}_k)\rho^{*}(\bm{x}_{k'})\rho(\bm{x}_{l'})\rho^{*}(\bm{x}_{l})\frac{\abs{C_i}^2 \abs{C_j}^2}{\abs{\a_i^r}^2\abs{\a_j^r}^2\abs{\a^t}^4}
  \end{aligned}\label{eq:ip3}
\end{equation}
where $\Phi_{i,j}^{k,k',l,l'}$ is as in~\eqref{eq:phasefunc}.

Then, under~\assref{assump:2}, we have that
\begin{equation}
  \begin{aligned}
  \frac{1}{M}\sum_{m=1}^M
  \mrm{e}^{-\mrm{i}\frac{\omega_m}{c_0}\Phi_{i,j}^{k,k',l,l'}}
  &=
  \frac{\mrm{e}^{-\mrm{i}\frac{\omega_c}{c_0}\Phi_{i,j}^{k,k',l,l'}}}{M}\sum_{m=1}^M\mrm{e}^{\mrm{i}\frac{B}{2c_0}\Phi_{i,j}^{k,k',l,l'}}\\
  &\qquad\qquad\times\mrm{e}^{-\mrm{i}\frac{(m-1)B}{Mc_0}\Phi_{i,j}^{k,k',l,l'}}\\
  &=
  \frac{\mrm{e}^{-\mrm{i}\frac{\omega_c'}{c_0}\Phi_{i,j}^{k,k',l,l'}}}{M}\frac{\sin\left(\frac{B}{2c_0}\Phi_{i,j}^{k,k',l,l'}\right)}{\sin\left(\frac{B}{2Mc_0}\Phi_{i,j}^{k,k',l,l'}\right)}\\
    &\approx \mrm{e}^{-\mrm{i}\frac{\omega_c'}{c_0}\Phi_{i,j}^{k,k',l,l'}}\opn{sinc}\left(\frac{B}{2c_0}\Phi_{i,j}^{k,k',l,l'}\right)
  \end{aligned}\label{eq:sinc}
\end{equation}
where $\omega_c' = \omega_c - \frac{B}{2M}$.  The second line is from geometric sum and the last
line is from small angle approximation.

Using~\eqref{eq:sinc} and changing the order of sum, and denoting $\abs{\alpha_{i,j}}^2 = \frac{\abs{C_i}^2 \abs{C_j}^2}{\abs{\a_i^r}^2\abs{\a_j^r}^2\abs{\a^t}^4}$, we have
\begin{equation}
  \begin{aligned}
    \norm{\mathcal{F}\tilde{\brho}}^2_2 &=
    \frac{1}{\binom{N}{2}}\sum_{i<j}\abs{\alpha_{i,j}}^2 \sum_{k,k',l,l'}\mrm{e}^{-\mrm{i}\omega_c'/c_0\Phi_{i,j}^{k,k',l,l'}}\\
    &\times\opn{sinc}\left(\frac{B}{2c_0}\Phi_{i,j}^{k,k',l,l'}\right)\tilde{\rho}(\bm{x}_k,\bm{x}_{k'})
  \tilde{\rho}(\bm{x}_{l'},\bm{x}_l).
  \end{aligned}\label{eq:ip4}
\end{equation}
We can split~\eqref{eq:ip4} into two parts as
\begin{equation}
  \begin{aligned}
  \norm{\mathcal{F}\tilde{\brho}}^2_2 &=
  \sum_{i<j}\frac{\abs{\alpha_{i,j}}^2}{\binom{N}{2}}\left(\norm{\brho}_2^4 + \sum_{k\neq k',l\neq
      l'}\mrm{e}^{-\mrm{i}\frac{\omega_c'}{c_0}\Phi_{i,j}^{k,k',l,l'}}\right.\\
    &\qquad\left.\times\opn{sinc}\left(\frac{B}{2c_0}\Phi_{i,j}^{k,k',l,l'}\right)\tilde{\rho}(\bm{x}_k,\bm{x}_{k'})\tilde{\rho}(\bm{x}_{l'},\bm{x}_{l})
  \right)\\
  &= \sum_{i<j}\frac{\abs{\alpha_{i,j}}^2}{\binom{N}{2}}\left(\norm{\tilde{\brho}}_F^2 +
    \opn{Re}\left\{\sum_{k\neq k',l\neq l'}\mrm{e}^{-\mrm{i}\frac{\omega_c'}{c_0}\Phi_{i,j}^{k,k',l,l'}}\right.\right.\\
      &\qquad\left.\left.\times\opn{sinc}\left(\frac{B}{2c_0}\Phi_{i,j}^{k,k',l,l'}\right)\tilde{\rho}(\bm{x}_k,\bm{x}_{k'})\tilde{\rho}(\bm{x}_{l'},\bm{x}_{l})\right\}
  \right)\\
  &= \sum_{i<j}\frac{\abs{\alpha_{i,j}}^2}{\binom{N}{2}}\left(\norm{\tilde{\brho}}_F^2 +
    \sum_{k\neq  k',l\neq l'}\opn{Re}\left\{\mrm{e}^{-\mrm{i}\frac{\omega_c'}{c_0}\Phi_{i,j}^{k,k',l,l'}}\right.\right.\\
      &\qquad\left.\left.\times\opn{sinc}\left(\frac{B}{2c_0}\Phi_{i,j}^{k,k',l,l'}\right)\tilde{\rho}(\bm{x}_k,\bm{x}_{k'})\tilde{\rho}(\bm{x}_{l'},\bm{x}_{l})\right\}
  \right).
  \end{aligned}\label{eq:splits}
\end{equation}
Having a real-valued $\tilde{\brho}$, we rewrite the latter term
in~\eqref{eq:splits} as
\begin{equation}
  \begin{aligned}
    W_{i,j} &= \sum_{k\neq k',l\neq
      l'}\opn{Re}\left\{\mrm{e}^{-\mrm{i}\frac{\omega_c'}{c_0}\Phi_{i,j}^{k,k',l,l'}}\opn{sinc}\left(\frac{B}{2c_0}\Phi_{i,j}^{k,k',l,l'}\right)\right.\\
      &\qquad\qquad\qquad\left.\times\tilde{\rho}(\bm{x}_k,\bm{x}_{k'})\tilde{\rho}(\bm{x}_{l'},\bm{x}_{l})\right\}\\
&=
\sum_{k\neq k',l\neq
  l'}\cos\omega_c'/c_0\Phi_{i,j}^{k,k',l,l'}\opn{sinc}\left(\frac{B}{2c_0}\Phi_{i,j}^{k,k',l,l'}\right)\\
&\qquad\qquad\qquad\times\tilde{\rho}(\bm{x}_k,\bm{x}_{k'})\tilde{\rho}(\bm{x}_{l'},\bm{x}_{l})\\
&= \sum_{k\neq k',l\neq
  l'}\mathcal{K}(\Phi_{i,j}^{k,k',l,l'})\tilde{\rho}(\bm{x}_k,\bm{x}_{k'})\tilde{\rho}(\bm{x}_{l'},\bm{x}_{l})\label{eq:splits2}.
\end{aligned}
\end{equation}
We can further rewrite $\mathcal{K}(\Phi_{i,j}^{k,k',l,l'})$~\eqref{eq:splits2} using trigonometric identity as
\begin{equation}
    \mathcal{K}(\Phi_{i,j}^{k,k',l,l'}) =\frac{\sin\left[\frac{\left(\omega_c'+\frac{B}{2}\right)\Phi_{i,j}^{k,k',l,l'}}{c_0}\right] -
    \sin\left[\frac{\left(\omega_c'-\frac{B}{2}\right)\Phi_{i,j}^{k,k',l,l'}}{c_0}\right]}{B\frac{\Phi_{i,j}^{k,k',l,l'}}{c_0}}
\end{equation}
which proves the claim.

\subsection{Proof of Proposition~\ref{prop:prop1}}\label{prf:prop}
First we express $\mathcal{K}$ as
\begin{equation}
  \begin{aligned}
  \mathcal{K}(\Phi_{i,j}^{k,k',l,l'}) &=\frac{\sin\left[\frac{\left(\omega_c'+\frac{B}{2}\right)\Phi_{i,j}^{k,k',l,l'}}{c_0}\right] -
    \sin\left[\frac{\left(\omega_c'-\frac{B}{2}\right)\Phi_{i,j}^{k,k',l,l'}}{c_0}\right]}{B\frac{\Phi_{i,j}^{k,k',l,l'}}{c_0}}\\
  &=
  \frac{\omega_c'}{B}\left\{s_1(\Phi_{i,j}^{k,k',l,l'}) - s_2(\Phi_{i,j}^{k,k',l,l'})\right\}\\
  &\qquad+ \frac{1}{2}\left\{s_1(\Phi_{i,j}^{k,k',l,l'}) + s_2(\Phi_{i,j}^{k,k',l,l'})\right\}
  \end{aligned}\label{eq:Ksplits}
\end{equation}
where
\begin{equation}
  \begin{aligned}
  s_1(\Phi_{i,j}^{k,k',l,l'}) &= \opn{sinc}\left[\frac{\left(\omega_c'+\frac{B}{2}\right)\Phi_{i,j}^{k,k',l,l'}}{c_0}\right]\\
    s_2(\Phi_{i,j}^{k,k',l,l'}) &= \opn{sinc}\left[\frac{\left(\omega_c'-\frac{B}{2}\right)\Phi_{i,j}^{k,k',l,l'}}{c_0}\right]
  \end{aligned}
\end{equation}
Given~\eqref{eq:Ksplits}, it suffices to prove that
\begin{align}
  \lim_{\omega_c'\rightarrow\infty}\frac{\omega_c'}{B}\opn{sinc}\left[\left(\omega_c'+\frac{B}{2}\right)\frac{\Phi_{i,j}^{k,k',l,l'}}{c_0}\right]
  &= \frac{c_0\pi}{B}\delta(\Phi_{i,j}^{k,k',l,l'})\label{eq:deltalim}\\
   \lim_{\omega_c'\rightarrow\infty}\frac{\omega_c'}{B}\opn{sinc}\left[\left(\omega_c'-\frac{B}{2}\right)\frac{\Phi_{i,j}^{k,k',l,l'}}{c_0}\right]
  &= \frac{c_0\pi}{B}\delta(\Phi_{i,j}^{k,k',l,l'}).\label{eq:deltalim2}
\end{align}
This can be proved using similar machinery to proving the delta function limit for sequence
  of scaled sinc functions.
 We prove the first equation~\eqref{eq:deltalim}.  The second equation follows similarly.

Let $f\in S(\mathbb{R})$ be a smooth test function, where $S(\mathbb{R})$ is the Schwartz space.
Then we need to prove that
 \begin{equation}
   \lim_{\omega_c'\rightarrow\infty}\int_{-\infty}^{\infty}g_{\omega_c'}(x)f(x)dx = \frac{c_0\pi}{B}f(0).
 \end{equation}
 where
 \begin{equation}
   g_{\omega_c'}(x) = \frac{\omega_c'}{B}\opn{sinc}\left[\left(\omega_c'+\frac{B}{2}\right)\frac{x}{c_0}\right]
 \end{equation}
 Let $\epsilon>0$ and break-up the integral into two parts.
 \begin{equation}
   \begin{aligned}
     \int_{-\infty}^{\infty}K(x)f(x)dx =& \int_{\abs{x}\geq \epsilon}g_{\omega_c'}(x)f(x)dx\\
     &+
   \int_{\abs{x}\leq \epsilon}g_{\omega_c'}(x)f(x)dx.
   \end{aligned}\label{eq:twoparts}
 \end{equation}
 The first part of the integral is
 \begin{equation}
   \begin{aligned}
     \int_{\abs{x}\geq \epsilon}g_{\omega_c'}(x)f(x)dx =& \int_{\epsilon}^{\infty}g_{\omega_c'}(x)f(x)dx\\
     &+ \int_{-\infty}^{-\epsilon}g_{\omega_c'}(x)f(x)dx
   \end{aligned}\label{eq:part1}
 \end{equation}
 We compute the first part of~\eqref{eq:part1}.  Integrating by parts,
 \begin{equation}
   \begin{aligned}
     &\int_{\epsilon}^{\infty}g_{\omega_c'}(x)f(x)dx =
     \frac{\frac{c_0}{B}\omega_c'}{\omega_c'+\frac{B}{2}}\left[\left.-\frac{f(x)\cos\left(\frac{\left(\omega_c'+ \frac{B}{2}\right)x}{c_0}\right)}{x\left(\omega_c'+\frac{B}{2}\right)}\right|_{\epsilon}^{\infty}\right.\\
       &\qquad\qquad\left.+
       \int_{\epsilon}^{\infty}\frac{xf'(x) - f(x)}{x^2}\frac{\cos\left(\left(\omega_c'
             + \frac{B}{2}\right)\frac{x}{c_0}\right)}{\left(\omega_c'+\frac{B}{2}\right)}
     \right]
   \end{aligned}\label{eq:intbypart}
 \end{equation}
 Since $f\in S(\mathbb{R})$, taking the limit as $\omega_c'\rightarrow\infty$,~\eqref{eq:intbypart}
 goes to $0$.  Similarly, we
 can see that the second integral in~\eqref{eq:part1} also goes to zero.  Thus,
 \begin{equation}
   \lim_{\omega_c'\rightarrow\infty}\int_{\abs{x}\geq \epsilon}g_{\omega_c'}(x)f(x)dx = 0
 \end{equation}
 Now, the second integral in~\eqref{eq:twoparts} can be rewritten as
 \begin{equation}
   \int_{-\epsilon}^{\epsilon} g_{\omega_c'}(x) (f(x)-f(0)) dx + f(0)\int_{-\epsilon}^{\epsilon}
   g_{\omega_c'}(x) dx
 \end{equation}
 The first integral can be rewritten as
 \begin{equation}
   \frac{c_0}{B}\frac{\omega_c'}{\omega_c'+\frac{B}{2}}\int_{-\epsilon}^{\epsilon}\sin\left(\left(\omega_c'+\frac{B}{2}
     \right)\frac{x}{c_0}\right) \frac{f(x)-f(0)}{x}dx\label{eq:intfirst}
 \end{equation}
 Since $f\in S(\mathbb{R})$, we can use \emph{Riemann-Lebesgue lemma} to conclude that~\eqref{eq:intfirst}
 goes to zero as $\omega_c'\rightarrow \infty$~\cite{rudin1987}.

For second integral, we have
 \begin{equation}
   \frac{c_0}{B}\frac{\omega_c'}{\omega_c'+\frac{B}{2}}f(0)\int_{-\epsilon}^{\epsilon}\frac{\sin\left(\left(\omega_c'+\frac{B}{2}
     \right)\frac{x}{c_0}\right)}{x}dx.
 \end{equation}
 After change of variables $u = \left(\omega_c'+\frac{B}{2}
 \right)\frac{x}{c_0}$, we have
 \begin{equation}
   \frac{c_0}{B}\frac{\omega_c'}{\omega_c'+\frac{B}{2}}f(0)\int_{-\epsilon\left(\omega_c'+\frac{B}{2}
 \right)/c_0}^{\epsilon\left(\omega_c'+\frac{B}{2}
 \right)/c_0} \frac{\sin u}{u}du \rightarrow
 \frac{c_0\pi}{B}f(0)
 \end{equation}
 as $\omega_c'\rightarrow\infty$.
 We can  use similar argument for where
 \begin{equation}
   g_{\omega_c'}(x) = \frac{\omega_c'}{B}\opn{sinc}\left[\left(\omega_c'-\frac{B}{2}\right)\frac{x}{c_0}\right]
 \end{equation}

\subsection{Proof of Proposition~\ref{prop:prop2}}\label{prf:prop2}
Without loss of generality, let the receivers and transmitters have common elevation angle $\phi$
such that
\begin{align}
  \hat{\a}_i^r &= [\cos\phi\cos\theta_i,\cos\phi\sin\theta_i,\sin\phi]^T\\
  \hat{\a}^t = &=[\cos\phi\cos\theta_t,\cos\phi\sin\theta_t,\sin\phi]^T.
\end{align}
where $\theta_i$ is the azimuth angle of the $i$-th receivers look-direction, $\theta_t$ is the
azimuth angle of the transmitter look-direction and $\phi$ is the elevation angle.
Furthermore, we have that for any $k$ and $k'$
\begin{equation}
  \bm{x}_k-\bm{x}_{k'} = \norm{\bm{x}_k-\bm{x}_{k'}}[\cos\theta_{k,k'},\sin\theta_{k,k'}]^T
\end{equation}
where $\theta_{k,k'}$ is the angle of the vector $\bm{x}_{k}-\bm{x}_{k'}$.
Then we have that
\begin{equation}
  \begin{aligned}
  \frac{\Phi_{i,j}^{k,k',l,l'}}{\cos\phi} =&
  \norm{\bm{x}_k-\bm{x}_{k'}}(\cos(\theta_i-\theta_{k,k'})+\cos(\theta_t-\theta_{k,k'}))\\
  &- \norm{\bm{x}_l-\bm{x}_{l'}}(\cos(\theta_j-\theta_{l,l'})+\cos(\theta_t-\theta_{l,l'})).
  \end{aligned}
\end{equation}
Thus, for the non-diagonal terms where $k\neq k'$, $l\neq l'$ we have that $\Phi_{i,j}^{k,k',l,l'} = 0$ if
\begin{equation}
  \begin{aligned}
  \frac{\norm{\bm{x}_k-\bm{x}_{k'}}}{\norm{\bm{x}_l-\bm{x}_{l'}}}&(\cos(\theta_i-\theta_{k,k'})+\cos(\theta_t-\theta_{k,k'}))\\
  &- \cos(\theta_t-\theta_{l,l'}) = \cos(\theta_j-\theta_{l,l'}).
  \end{aligned}\label{eq:Phi0set}
\end{equation}
For fixed $k,k',l,l'$ and $i$, there are at most $2$ values of $\theta_j$'s for
which~\eqref{eq:Phi0set} is satisfied.
Furthermore, we know that $\abs{\alpha_{i,j}}^2$'s must be bounded.  Thus, by~\propref{prop:prop1}, for
each fixed $k,k',l,l'$
where $k\neq k'$ and $l\neq l'$ we have that
\begin{equation}
  \frac{1}{\binom{N}{2}}\sum_{i<j}\abs{\alpha_{i,j}}^2 \lim_{\omega_c'\rightarrow\infty}
  \mathcal{K}(\Phi_{i,j}^{k,k',l,l'}) = \mathcal{O}\left(\frac{1}{N}\right).
\end{equation}
Now taking the limit as $N\rightarrow\infty$, we have the desired result.

\subsection{Proof of Theorem~\ref{thm:Theorem1}}\label{prf:Theorem1}
We want to upper bound the following
\begin{equation}
  \begin{aligned}
  \abs{\frac{\sum_{i,i<j}\abs{\alpha_{i,j}}^2W_{i,j}}{\binom{N}{2}}}\approx&
  \left|\frac{\alpha}{\binom{N}{2}}\sum_{i,i<j}\sum_{k,k\neq k',l,l\neq
      l'}\mathcal{K}(\Phi_{i,j}^{k,k',l,l'})\right.\\
    &\left.\times\tilde{\rho}(\bm{x}_k,\bm{x}_{k'})\tilde{\rho}(\bm{x}_{l'},\bm{x}_{l})\right|.
  \end{aligned}\label{eq:Wij}
\end{equation}
where, $\abs{\alpha_{i,j}}^2 \approx \alpha = \frac{\abs{C}^4}{(\bm{a}_1^r)^4(\bm{a}^t)^2}$ by
\assref{assump:config}.  Without loss of generality we set $\alpha = 1$.
We begin by noting that
\begin{equation}
  \Phi_{i,j}^{k,k',l,l'} = (\hat{\a}_i + \hat{\a}_t)\cdot (\bm{x}_k - \bm{x}_{k'})  + \beta_{j}^{l,l'}.\label{eq:phiij-b}
\end{equation}
where
\begin{equation}
  \beta_j^{l,l'} = -(\hat{\a}_j + \hat{\a}_t)\cdot (\bm{x}_l - \bm{x}_{l'}).
\end{equation}
Thus, fixing $l$, $l'$, and $k$, we have convolution between $\mathcal{G}$ and $\rho$.  Let
\begin{equation}
  \mathcal{G}_{i,j}(\bm{x}_k-\bm{x}_{k'}) = \mathcal{K}((\hat{\a}_i + \hat{\a}_t)\cdot (\bm{x}_k -
  \bm{x}_{k'}) + \beta_j^{l,l'}).
\end{equation}
We take the Fourier Transform of $\mathcal{G}_{i,j}$ and $\rho$ to represent the convolution.  Denoting,
$\hat{\mathcal{G}}_{i,j}$ as the Fourier Transform of $\mathcal{G}_{i,j}$, we have
\begin{equation}
  \begin{aligned}
  &\sum_{k\neq k'}\sum_{l\neq
    l'}\mathcal{K}(\Phi_{i,j}^{k,k',l,l'})\tilde{\rho}(\bm{x}_k,\bm{x}_{k'})\tilde{\rho}(\bm{x}_{l'},\bm{x}_{l})\\
  &= \frac{1}{4\pi^2}\sum_{l\neq l'} \tilde{\rho}(\bm{x}_l,\bm{x}_{l'})\sum_k \rho(\bm{x}_k)\int
  \mrm{e}^{\mrm{i}\bm{\omega}\cdot\bm{x}_k}\hat{\mathcal{G}}_{i,j}(\bm{\omega})\hat{\rho}(\bm{\omega})
  d\bm{\omega}
  \end{aligned}\label{eq:FT1}
\end{equation}

To compute $\hat{\mathcal{G}}_{i,j}$, we first rewrite $\mathcal{G}_{i,j}$ as
\begin{equation}
  \begin{aligned}
  &\mathcal{G}_{i,j}(\bm{x}_k) =
  \frac{\omega_c'+\frac{B}{2}}{B}\opn{sinc}\left(\frac{\omega_c'+\frac{B}{2}}{c_0}((\hat{\a}_i +
    \hat{\a}_t)\cdot\bm{x}_k + \beta_j^{l,l'})\right) \\&\qquad-
  \frac{\omega_c'-\frac{B}{2}}{B}\opn{sinc}\left(\frac{\omega_c'-\frac{B}{2}}{c_0}((\hat{\a}_i +
    \hat{\a}_t)\cdot\bm{x}_k + \beta_j^{l,l'})\right).
  \end{aligned}\label{eq:K2}
\end{equation}
Let $\bm{x}_k = [x_1^k,x_2^k]^T$, $\bm{\omega} = [\omega_1,\omega_2]^T$ and $\theta_i$ be the azimuth
angle of the $i$-th receiver's look-direction.  Then, given~\eqref{eq:K2},
the Fourier Transform of $\mathcal{G}$ and using the assumption that $\hat{\bm{a}}_t = [1,0]^T$,
\begin{equation}
  \hat{\mathcal{G}}_{i,j}(\bm{\omega}) =
  \frac{c_0}{B}\frac{K}{L}\mrm{e}^{\mrm{i}\omega_1\gamma_{i,j}^{l,l'}}S(\bm{\omega})R(\omega_1)\label{eq:Ghat}
\end{equation}
where
\begin{equation}
  \gamma_{i,j}^{l,l'} = \frac{\beta_j^{l,l'}}{\cos\phi(\cos\theta_i + 1)} 
  \label{eq:gammaij}
\end{equation}
\begin{equation}
  \begin{aligned}
    R(\omega_1) =& \frac{1}{\cos\phi(\cos\theta_i+1)}\left[
      \opn{rect}\left(\frac{\omega_1}{2\frac{(\omega_c'+\frac{B}{2})\cos\phi(\cos\theta_i+1)}{c_0}}\right)\right.\\
    &\qquad\left.-
      \opn{rect}\left(\frac{\omega_1}{2\frac{(\omega_c'-\frac{B}{2})\cos\phi(\cos\theta_i+1)}{c_0}}\right)\right]
  \end{aligned}\label{eq:Ro}
\end{equation}
and
\begin{equation}
  S(\bm{\omega}) = \opn{sinc}\left(\left(\omega_2
  - \omega_1\frac{\sin\theta_i}{\cos\theta_i+1}\right)\frac{L}{2}\right).\label{eq:So}
\end{equation}
Noting that $R$ is only non-zero where $\frac{(\omega_c'-B/2)}{c_0}\cos\phi(\cos\theta_i +1)\leq
\omega_1\leq \frac{(\omega_c'+B/2)}{c_0}\cos\phi(\cos\theta_i +1)$, and $\omega_c'\gg B/2$, we
approximate~\eqref{eq:gammaij} as
\begin{equation}
    \omega_1\gamma_{i,j}^{l,l'}\approx \frac{\omega_1[\cos\theta_j,\sin\theta_j]^T\cdot\bm{x}_l-\bm{x}_{l'}}{\cos\theta_i + 1} + \frac{\omega_c'(x_1^l-x_1^{l'})\cos\phi}{c_0}.
\end{equation}
Next, we note that
\begin{equation}
  \sum_k\rho(\bm{x}_k)\mrm{e}^{\mrm{i}\bm{\omega}\cdot\bm{x}_k} = \hat{\rho}^*(\bm{\omega}).
\end{equation}
Thus, interchanging the sum and the integral in~\eqref{eq:FT1}, and plugging in~\eqref{eq:Ghat} we have
\begin{align}
  \sum_{k\neq k'}\sum_{l\neq
  l'}&\mathcal{K}(\Phi_{i,j}^{k,k',l,l'})\tilde{\rho}(\bm{x}_k,\bm{x}_{k'})\tilde{\rho}(\bm{x}_{l'},\bm{x}_{l})\nonumber\\
  &= \frac{1}{4\pi^2}\frac{c_0}{B}\frac{K}{L}\int
    S(\bm{\omega})R(\omega_1)\abs{\hat{\rho}(\bm{\omega})}^2\nonumber\\
  &\qquad\qquad\quad\times\sum_{l\neq l'}
    \mrm{e}^{\mrm{i}\omega_1\gamma_{i,j}^{l,l'}}\tilde{\rho}(\bm{x}_l,\bm{x}_{l'}) d\bm{\omega}\\
  &= \frac{1}{4\pi^2}\frac{c_0}{B}\frac{K}{L}\int
    S(\bm{\omega})R(\omega_1)\abs{\hat{\rho}(\bm{\omega})}^2\nonumber\\
  &\qquad\qquad\qquad\times\abs{\hat{\rho}(\bm{\omega}')}^2 d\bm{\omega}
    \label{eq:FT2}
\end{align}
where
\begin{align}
  \bm{\omega}' &= \omega_1'[\cos\theta_j,\sin\theta_j]^T + \frac{\omega_c'}{c_0}\cos\phi[1,0]^T\\
  \omega_1' &= \frac{\omega_1}{\cos\theta_i+1}
\end{align}

Now, by employing Cauchy-Schwartz, we have
\begin{equation}
  \begin{aligned}
  &\abs{\int
    S(\bm{\omega})R(\omega_1)\abs{\hat{\rho}(\bm{\omega})}^2 \abs{\hat{\rho}(\bm{\omega}')}^2
    d\bm{\omega}}\\
  &\qquad\leq
  \sqrt{\int S^2(\omega_1',\omega_2)\tilde{R}^2(\omega_1')\abs{\hat{\rho}(\bm{\omega}')}^4d\omega_2
    d\omega_1'}\\
  &\qquad\qquad\qquad\times\sqrt{\int\abs{\hat{\rho}(\bm{\omega})}^4d\bm{\omega}}
  \end{aligned}\label{eq:CS1}
\end{equation}
where
\begin{equation}
  \tilde{R}(\omega_1') = (\cos\theta_i+1)R((\cos\theta_i+1)\omega_1').
\end{equation}
By Jensen's inequality,~\eqref{eq:CS1} becomes
\begin{align}
    &\abs{\int
      S(\bm{\omega})R(\omega_1)\abs{\hat{\rho}(\bm{\omega})}^2 \abs{\hat{\rho}(\bm{\omega}')}^2
  d\bm{\omega}}\nonumber\\
  &\leq
                     \sqrt{\int S^2(\omega_1',\omega_2)\tilde{R}^2(\omega_1')\abs{\hat{\rho}(\bm{\omega}')}^4d\omega_2
                     d\omega_1'}                 \int
                     \abs{\hat{\rho}(\bm{\omega})}^2d\bm{\omega}\\
                   &=4\pi^2\norm{\bm{\rho}}^2_2\sqrt{\int
                     S^2(\omega_1',\omega_2)\tilde{R}^2(\omega_1') \abs{\hat{\rho}(\bm{\omega}')}^4d\omega_2
                     d\omega_1'}
                     \label{eq:CS2}
\end{align}

Noting that for any fixed $\omega_1$,
\begin{equation}
  \int S^2(\bm{\omega})d\omega_2 = \frac{2\pi}{L},
\end{equation}
we have
\begin{equation}
  \begin{aligned}
  &\int S^2(\omega_1',\omega_2)\tilde{R}^2(\omega_1')
  \abs{\hat{\rho}(\bm{\omega}')}^4d\omega_2d\omega_1'\\
  &\qquad= \frac{2\pi}{L}\int
  \tilde{R}^2(\omega_1')\abs{\hat{\rho}(\bm{\omega}')}^4d\omega_1'.
  \end{aligned}\label{eq:CS3}
\end{equation}
We use Jensen's inequality once more to get
\begin{equation}
  \begin{aligned}
  &\sqrt{\int S^2(\omega_1',\omega_2)\tilde{R}^2(\omega_1')
    \abs{\hat{\rho}(\bm{\omega}')}^4d\omega_2d\omega_1'}\\
  &\qquad\leq \sqrt{\frac{2\pi}{L}}\int
\tilde{R}(\omega_1')\abs{\hat{\rho}(\bm{\omega}')}^2d\omega_1'.
\end{aligned}\label{eq:CS4}
\end{equation}

Next, approximating the sum over $\theta_j$ as an integral, we have
\begin{equation}
  \begin{aligned}
  &\frac{1}{2\binom{N}{2}}\sum_{i} \frac{N}{A} \sum_{i\neq
    j} \int \frac{A}{N} \tilde{R}(\omega_1')\abs{\hat{\rho}(\bm{\omega}')}^2d\omega_1'\\
  &\qquad\approx
\frac{1}{A} \left( \int
\tilde{R}(\omega_1')\abs{\hat{\rho}(\bm{\omega}')}^2 d\omega_1' d\theta_j + E_R\right),
\end{aligned}\label{eq:CS5}
\end{equation}
where $E_R$ denotes the Riemann sum error, $A$ is the aperture of look directions in the imaging setup.

We first consider the inner integration over $\omega_1'$, where $\bm{\omega}' = [\omega_1', \theta_j]$. Using Cauchy-Schwartz and Jensen's
inequalities,
\begin{align}
  &\int\displaylimits_A \left[ \int \frac{1}{\abs{ \omega_1' }} \tilde{R}(\omega_1') \abs{ \omega_1'} \abs{\hat{\rho}(\bm{\omega}')}^2 d{\omega_1'} \right] d\theta_j \nonumber\\
  &\qquad\leq \int\displaylimits_A \left[\sqrt{\int
    \frac{1}{(\omega_1')^2} \tilde{R}^2(\omega_1') d{\omega_1'}}\int \abs{ \omega_1'} \abs{\hat{\rho}(\bm{\omega}')}^2 d{\omega_1'} \right] d\theta_j \nonumber\\
  &\qquad = \sqrt{\int
    \frac{1}{(\omega_1')^2} \tilde{R}^2(\omega_1') d{\omega_1'}} \int\displaylimits_A \int \abs{ \omega_1'} \abs{\hat{\rho}(\bm{\omega}')}^2 d\bm{\omega}'.\label{eq:CS7}
\end{align}

Computing the first integral in~\eqref{eq:CS7}, we get
  \begin{align}
    \int\frac{1}{(\omega_1')^2} \tilde{R}^2(\omega_1') d\omega_1' &=
    \frac{2}{\cos^2\phi}\int_{\frac{\omega_c'-B/2}{c_0}\cos\phi}^{\frac{\omega_c'+B/2}{c_0}\cos\phi}\frac{1}{(\omega_1')^2}d\omega_1' \nonumber\\
    &= \frac{2c_0}{\cos^3\phi}\left(\frac{1}{\omega_c'-B/2} - \frac{1}{\omega_c'+B/2}\right)\nonumber\\
    &= \frac{2Bc_0}{\cos^3\phi((\omega'_c)^2 - (B/2)^2)}.\label{eq:CS8}
  \end{align}
Consider the second integral in~\eqref{eq:CS7}. Since the integrand is strictly positive, from the $\theta_j$ integration we have
\begin{equation}
  \int\displaylimits_A \int \abs{ \omega_1'} \abs{\hat{\rho}(\bm{\omega}')}^2 d\bm{\omega'} \leq \int\displaylimits_{2\pi} \int \abs{ \omega_1'} \abs{\hat{\rho}(\bm{\omega}')}^2 d\bm{\omega'}. \label{eq:CS9}
\end{equation}
We now make the following change of variables
\begin{equation}
  \cos\theta_j\omega_1' = \omega_1'', \ \ \sin\theta_j\omega_1' = \omega_2''.
\end{equation}
Computing the Jacobian, we get
\begin{equation}
  J =  \frac{1}{\abs{\omega_1'(\bm{\omega}'')}} = \frac{1}{\sqrt{(\omega_1'')^2 +  (\omega_2'')^2}}.
\end{equation}
Thus, setting $\bm{\omega}'' = [\omega_1'', \omega_2'']$, the upper bound in~\eqref{eq:CS9} becomes
\begin{equation}
\int \abs{ \omega_1'} \abs{\hat{\rho}(\bm{\omega}')}^2 d\bm{\omega'} = \int \abs{\hat{\rho}(\bm{\omega}'')}^2 d\bm{\omega}'' = 4\pi^2 \| \brho \|^2,
\label{eq:CS6}
\end{equation}
where the last identity follows from Parseval's theorem.
Hence, we obtain the upper bound on our integral approximation, i.e., the first term in \eqref{eq:CS5} as
\begin{equation}
\begin{split}
\int \tilde{R}(\omega_1')\abs{\hat{\rho}(\bm{\omega}')}^2 \bm{\omega}'  & \leq  \frac{4\pi^2 \| \brho \|^2   \sqrt{2Bc_0}}{(\cos \phi)^{3/2}  \sqrt{(\omega'_c)^2 - (B/2)^2}} \\
 & \approx \frac{4\pi^2 \| \brho \|^2   \sqrt{2Bc_0}}{ \omega'_c (\cos \phi)^{3/2}}. 
\label{eq:CS10}
\end{split}
\end{equation}

We next evaluate the error term of the integral approximation in~\eqref{eq:CS6}. Namely, the midpoint Riemann-sum approximation error is upper bounded as
\begin{equation}\label{eq:RieErrDef}
| E_R | \leq Q \frac{A^3}{24 N^2}
\end{equation}
where $Q$ is defined as
\begin{equation}\label{eq:RieErr_1}
Q = \underset{\theta}{\mathrm{max}}  \abs{\frac{\partial^2}{\partial\theta^2} \int_{\abs{\frac{\omega_c-B/2}{c_0}\cos\phi}}^{\abs{\frac{\omega_c+B/2}{c_0}\cos\phi}} \abs{\hat{\rho}(\bm{\omega}')}^2 d\omega_1'}
\end{equation}
using the definition of $\tilde{R}(\omega_1')$.

Since the integrand is the squared absolute value of the Fourier transform of the reflectivity function evaluated at frequencies $\bm{\omega}'$, \eqref{eq:RieErr_1} can equivalently be written as
\begin{equation}
Q = \underset{\theta}{\mathrm{max}} \abs{\sum_{l, l'} \frac{\partial^2}{\partial \theta^2} \tilde{f}(\theta, \bm{x}_l, \bm{x}_{l'}) \tilde{\rho}(\bm{x}_l, \bm{x}_{l'})} 
\end{equation}
which is the form of a Frobenius inner product, where
\begin{equation}\label{eq:neweq1}
\tilde{f}(\theta, \bm{x}_l, \bm{x}_{l'}) := \int_{\abs{\frac{\omega_c-B/2}{c_0}\cos\phi}}^{\abs{\frac{\omega_c+B/2}{c_0}\cos\phi}} \mathrm{e}^{-\mathrm{i} \bm{\omega}' \cdot (\bm{x}_l - \bm{x}_{l'})} d\omega_1'.
\end{equation}

Observe that $Q$ has an upper bound that only depends on the $\ell_2$ norm of the underlying scene as
\begin{equation}\label{eq:neweq2}
Q \leq \underset{\theta}{\mathrm{max}} \| \frac{\partial^2}{\partial \theta^2} \tilde{\f} \|_F \| {\brho} \|^2,
\end{equation}
and that $\underset{\theta}{\mathrm{max}} \| \frac{\partial^2}{\partial \theta^2}  \tilde{\f} \|_F$ has only dependence on the imaging system, hence yields a universal upper bound for any scene reflectivity function in the specific imaging geometry. 
Evaluating the integral in~\eqref{eq:neweq1} and having $\omega_c \gg B$, $\tilde{f}$ can be approximated under the narrow-band assumption as 
\begin{equation}
\tilde{f} \approx \frac{B}{c_0} \mathrm{sinc}\left(\frac{\omega_c}{c_0} \gamma(\theta) \cdot (\bm{x}_l - \bm{x}_{l'})\right).
\end{equation}
Denoting $z = \frac{\omega_c}{c_0} \gamma(\theta) \cdot (\bm{x}_l - \bm{x}_{l'})$, the second order derivative of each component in $\tilde{f}$ with respect to $\theta$ is obtained as
\begin{equation}\label{eq:neweq5}
\frac{\partial^2 z}{\partial \theta^2} \left( \frac{\cos z}{z} - \frac{\sin z}{z^2}\right) +  \left(\frac{\partial z}{\partial \theta}\right)^2 \left( - \frac{\sin z}{z} - \frac{2 \cos z}{z^2} + \frac{2 \sin z}{z^3} \right)
\end{equation}
where
\begin{equation}\label{eq:neweq6}
\frac{\partial^2 z}{\partial \theta^2} = -  \frac{\omega_c}{c_0} [\cos\theta, \sin\theta]^T \cdot (\bm{x}_l - \bm{x}_{l'}),
$$
$$
(\frac{\partial z}{\partial \theta})^2 =  \left(\frac{\omega_c}{c_0}\right)^2 ([-\sin\theta, \cos\theta]^T \cdot (\bm{x}_l - \bm{x}_{l'}) )^2.
\end{equation}

By definition, the amplitude of sinc-function derivatives in \eqref{eq:neweq5} have a decay rate of $1/z$, whereas $z$ and its derivatives in \eqref{eq:neweq5} grow with an order of the scene dimension $L$, as $\mathcal{O}(\frac{\omega_c}{c_0} L)$.
Hence, evaluating the squared integral for the Frobenius norm, the growth is at a maximal order of $\mathcal{O}\left((\frac{\omega_c}{c_0})^3 L^3 \right)$. We thereby obtain a final approximate upper bound on the universal constant $Q$ as
\begin{equation}\label{eq:ErrBd}
\begin{split}
 \abs{\frac{\partial^2}{\partial\theta^2}\int_{\abs{\frac{\omega_c-B/2}{c_0}\cos\phi}}^{\abs{\frac{\omega_c+B/2}{c_0}\cos\phi}}\hat{\rho}(\omega(\cos\theta+1), \omega\sin\theta)
    d\omega} \\
   \leq \mathcal{O}\left(\frac{B}{c_0} (\frac{\omega_c}{c_0})^{\frac{3}{2}} L^{\frac{3}{2}}\right) \| {\brho} \|^2.
\end{split}
\end{equation}
Putting the terms derived in~\eqref{eq:ErrBd} and ~\eqref{eq:CS10} into~\eqref{eq:CS4} and~\eqref{eq:FT2}, we have
\begin{equation}
\delta \leq \frac{c_0}{B} \frac{K}{L} \sqrt{\frac{2 \pi}{L}} \frac{1}{A} \left(\frac{4\pi^2 \sqrt{2Bc_0}}{ \omega'_c (\cos \phi)^{3/2}} + \mathcal{O}\left(\frac{B}{c_0} (\frac{\omega_c}{c_0})^{\frac{3}{2}} L^{\frac{3}{2}} \frac{A^3}{N^2}\right) \right)
\end{equation}
for the definition of $\delta$ in Definition \ref{def:RIP}, using the fact that $\| \brho \|^4 = \| \tilde{\brho} \|_F^2$ for the rank-1, PSD element $\tilde{\brho} = \brho \brho^H$.

Re-organizing the terms using \eqref{eq:defsthm1} and \eqref{eq:defsthm2} defined in the statement of Theorem \ref{thm:Theorem1}, and the fact that
\begin{equation}
  \frac{K}{L\sqrt{L}} = \frac{\sqrt{L}}{\Delta^2},
\end{equation}
we finally obtain
\begin{equation}
\delta \leq \frac{2\pi}{A} \frac{2 \lambda_c \sqrt{L \Delta_{res}}}{ \Delta^2 (\cos \phi)^{3/2} } + \mathcal{O}\left(\frac{K}{(N/A)^2} {\lambda_c^{-{3}/{2}}}  \right) ,
\end{equation}
which completes the proof.

\end{document}